\documentclass[11pt]{article}
\usepackage[margin=1in]{geometry}

\usepackage{authblk}

\usepackage{amsthm}
\usepackage{thm-restate}
\usepackage{lscape}
\usepackage{tabularx}
  
\usepackage{amsthm}
\usepackage{thmtools} 

\usepackage{thm-restate}
\usepackage{amsmath, amsfonts, mathtools}
\usepackage{braket}
\usepackage{graphicx}
\usepackage{amssymb}
\usepackage{dsfont}

\usepackage{color,soul}
\usepackage{makecell}
\usepackage{subfigure}
\usepackage{qcircuit}
\usepackage{booktabs}
\usepackage[scale=2]{ccicons}

\usepackage[dvipsnames]{xcolor}
\usepackage{hyperref}
\hypersetup{
    colorlinks = true,
    linkcolor = blue
    }
\usepackage[noabbrev,capitalize,nameinlink]{cleveref}

\usepackage{pgfplots}
\pgfplotsset{compat=1.17}
\usepgfplotslibrary{dateplot}

\usepackage{doi}
\usepackage[style=chem-angew,sorting=none,doi=true,url=false,terseinits=true, firstinits=true,giveninits=true,hyperref]{biblatex}

\DeclareNameAlias{default}{last-first}

\usepackage{xspace}
\usepackage{algorithm}
\usepackage{algpseudocode}

\newcommand{\kron}{\otimes}
\newcommand{\identity}{{\boldsymbol 1}}
\newcommand{\norm}[1]{\left\lVert#1\right\rVert}
\newcommand{\ketbra}[2]{\ket{#1}\!\bra{#1}}
\DeclarePairedDelimiter\ceil{\lceil}{\rceil}
\DeclarePairedDelimiter\floor{\lfloor}{\rfloor}

\newtheorem{theorem}{Theorem}[section]

\newtheorem{lemma}{Lemma}[section]
\newtheorem{fact}[theorem]{Fact}

\newtheorem{define}[theorem]{Definition}

\newtheorem{corr}{Corollary}[section]

\usepackage{svg}

\usepackage[draft]{changes} %

\title{Leveraging Hamiltonian Simulation Techniques \\to Compile Operations on Bosonic Devices}

\date{\today}

\author[1, 2, 3]{Christopher Kang$^*$}%
\author[4,5,6,7]{Micheline B.~Soley\thanks{These two authors contributed equally}}

\author[11, 8]{Eleanor Crane}
\author[6,12]{Steven M. Girvin}
\author[2, 9, 10]{Nathan Wiebe}

\affil[1]{School of Computer Science, University of Washington, Seattle, WA 98195 USA}
\affil[2]{Pacific Northwest National Laboratory, Richland, WA 99354 USA}
\affil[3]{Department of Computer Science, University of Chicago, Chicago, IL 60637 USA}
\affil[4]{Department of Chemistry, University of Wisconsin-Madison, Madison, WI 53706 USA}
\affil[5]{Department of Physics, University of Wisconsin-Madison, Madison, WI 53706 USA}
\affil[6]{Yale Quantum Institute, Yale University, New Haven, CT 06520 USA 
}
\affil[7]{Department of Chemistry, Yale University, New Haven, CT 06511 USA}
\affil[8]{Joint Quantum Institute \& Joint Center for Quantum Information and Computer Science, NIST/University of Maryland, College Park, MD 20742 USA}
\affil[9]{Department of Computer Science, University of Toronto, ON M5G 1Z7 Canada}
\affil[10]{Canadian Institute For Advanced Research/Institut Canadien de Recherches Avanc\'ees, ON M5G 1Z7 Canada}
\affil[11]{Research Laboratory of Electronics, Massachusetts Institute of Technology, Cambridge, MA 02139 USA}
\affil[12]{Department of Physics, Yale University, New Haven, CT 06511 USA}

\newcommand{\change}[1]{{#1}}

\newcommand{\qq}{qumode-qubit}
\newcommand{\qumode}{qumode}
\newcommand{\qumodes}{{\qumode}s}

\newcommand{\bch}{\mathrm{BCH}}
\newcommand{\trotter}{\mathrm{Trotter}}

\begin{document}
\date{}

\maketitle

\begin{abstract}
    Circuit quantum electrodynamics enables the combined use of qubits and oscillator modes. Despite a variety of available gate sets, many hybrid qubit-boson (i.e., qubit-oscillator) operations are realizable only through optimal control theory, which is oftentimes intractable and uninterpretable. We introduce an analytic approach with rigorously proven error bounds for realizing specific classes of operations via two matrix product formulas commonly used in Hamiltonian simulation, the Lie--Trotter--Suzuki and Baker--Campbell--Hausdorff product formulas. We show how this technique can be used to realize a number of operations of interest, including polynomials of annihilation and creation operators, namely $(a)^p (a^\dagger)^q$ for integer $p, q$.  We show examples of this paradigm including obtaining universal control within a subspace of the entire Fock space of an oscillator, state preparation of a fixed photon number in the cavity, simulation of the Jaynes--Cummings Hamiltonian,  and simulation of the Hong-Ou-Mandel effect.  This work demonstrates how techniques from Hamiltonian simulation can be applied to better control hybrid qubit-boson devices.
\end{abstract}

\newpage
\tableofcontents

\newpage
\section{Introduction}

Today, many quantum computing architectures are homogeneous, with the same type of qubit used throughout the device. From devices made of superconducting qubits \cite{arute2019quantum,bravyi2022future,reagor2018demonstration} to ion trap qubits \cite{wright2019benchmarking}, prior work largely focuses on linking qubits of the same type together in fault-tolerant ways. However, there is emerging work \cite{C2QA_ISA, crane2024hybrid, chakram2021seamless, stavenger2022bosonic,stein2023microarchitectures} 
studying the use of heterogeneous quantum computers that leverage two or more types of quantum architectures (e.g., qubits and oscillator modes). Heterogeneous devices hold promise because they can be tailored for specific physical simulation problems, which would be especially useful in applications like material discovery~\cite{Holstein_Knorzer_2022}, molecular simulation~\cite{Wang2020FCFs,WangConicalIntersection}, topological models~\cite{PhysRevB.98.174505} or lattice gauge theories~\cite{C2QA_LGT}.

In particular, hybrid \qq~models \cite{blais2021circuit} hold some advantages: for example, microwave qumodes have long lifetimes and large accessible Hilbert spaces, making them
attractive targets for quantum error correction~\cite{blais2021circuit}. 
Introducing \qumodes~also enables new physical gates, 
such as the M{\o}lmer-S{\o}rensen gate \cite{molmer-sorensen-gate} while at the same time enabling new forms of transduction between qubits and qumodes \cite{Boissonneault_Dispersive_2009}.
Oscillator interactions have unique features, like nonlinearities, which are challenging to simulate even with homogeneous quantum architectures~\cite{stavenger2022bosonic}.

Efficiently compiling logical operations to physical pulses is a critical, but computationally expensive task. 
In theory, pulse design techniques like optimal control theory (OCT)~\cite{,werschnik2007quantum} 
can produce pulses that implement arbitrary quantum transformations on a hybrid \qq~system. These techniques have been applied to a variety of physical systems, including NMR~\cite{khaneja2005optimal}, superconducting transmon qubits~\cite{Koch-OCT-2017}, and \qumodes~\cite{ozguler_numerical_2022,anders_petersson_optimal_2022,ma_quantum_2021}. In practice, OCT is computationally intensive and inflexible, requiring pulses to be recompiled on a case-by-case basis. Furthermore, OCT is almost always uninterpretable, yielding only a pulse which performs the desired operation without providing any physical intuition. 
Experimentally, these limitations prevent the high-fidelity realization of quantum algorithms; theoretically, the complexity of our quantum circuits often inappropriately ignores the classical cost of required compilation.

Inspired by recent experimental progress \cite{SNAP-PhysRevLett.115.137002,eickbusch2021fast}, we introduce an extensible control scheme for a universal, hybrid \qq~quantum computer (\cref{tab:all-formulas}).  
Specifically, we show how block-encoded operators can be manipulated 
using the Lie--Trotter--Suzuki and Baker--Campbell--Hausdorff matrix product formulas. We thus enable the creation of instruction set architectures (ISAs) that 
can be analytically compiled to experimentally available gate sets.  Prior art, namely \textcite{jacobs_engineering_2007}, has studied similar techniques to compile operations; we generalize these techniques to work in a variety of domains, including in settings with multiple \qumodes~and more exotic operators, and prove concrete error bounds on these techniques.  

We develop two parallel approaches, one which primarily uses the creation and annihilation operators which we refer to as based on `Fock methods,' and the other primarily relying on position and momentum operators which we refer to as based on `phase-space methods.' We demonstrate that both methods can be used to generate an ISA for \qq~devices. \change{Because our methods operate in both Fock and phase-space, we can achieve transformations that are natively described in either picture; for example, {exponentials of polynomials of annihilation and creation operators for Hamiltonian simulation and state preparation in Fock space and controlled parity and beam splitter operators in phase space}}. Our methods obtain almost-linear asymptotic scaling.

Furthermore, we use the previously mentioned formulas to realize a number of operations of interest, including polynomials of annihilation and creation operators, namely $a^p {a^\dagger}^q$ for integer $p, q$. These block-encoded operations are crucial for quantum signal processing (QSP) and certain problems in quantum simulation. Finally, we give examples for the Hamiltonian of a nonlinear material and applications to key unsolved problems in quantum simulation such as the Fermi-Hubbard model. 
While these approaches are expensive in terms of raw gate counts, because they are analytic they provide intuition into synthesizing the gate robustly. Furthermore, these gate sequences can be used as a starting point for optimal control methods, helping to avoid cold start issues.

\section{Preliminaries}
\begin{figure}
    \centering
    \includegraphics[width=\textwidth]{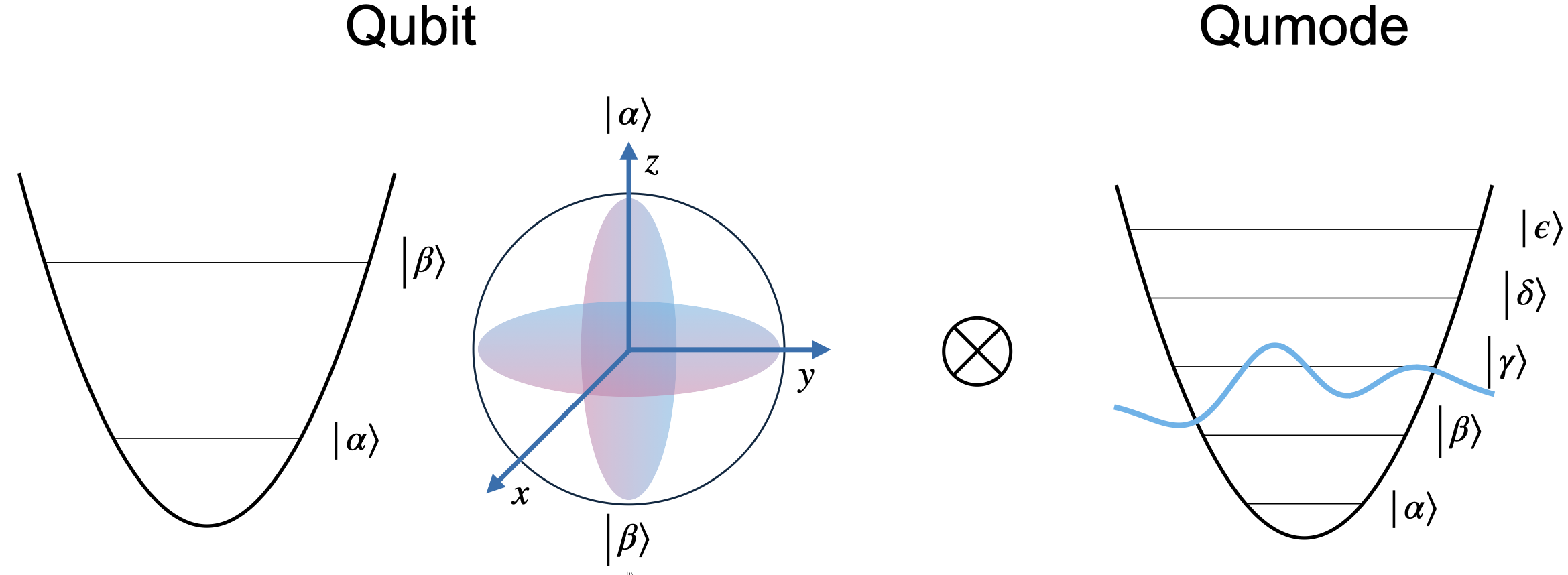}
    \caption{A Qubit-Qumode Hilbert space. On the left, the qubit is a typical 2 dimensional Hilbert space. On the right, the qumode consists of Fock states of value at most $\Lambda = 4$. }
    \label{fig:device-picture}
\end{figure}

{In this section, we introduce the hybrid \qq~architecture we operate on and the matrix product formulas we will use.}

We aim to provide a generic toolbox to build unitary transformations on hybrid \qq~devices. Such devices are common in quantum systems, spanning circuit quantum electrodynamics (superconducting qubits coupled to microwave photons) to ion-trap quantum computing (for which the mechanical modes of oscillation are coupled to atomic qubits).  The challenge is that fundamentally different insights are needed to compile unitaries in the hybrid setting over typical binary-based approaches.

We review the mathematical properties of \qumode~quantum mechanics that are needed to understand the basic operations considered for the ISA architecture that we consider.  
Specifically, we present an analytic instruction set architecture (ISA) based on the Lie--Trotter--Suzuki (Trotter) and Baker--Campbell--Hausdorff (BCH) decompositions for decomposition of gates of the form $U=e^{i\hat{H}\sigma^{j}}$, where $\hat{H}$ is a Hermitian
operator composed of phase-space operators and Pauli gates. Before jumping into the specific details of our gate operations, we need to review the basics of qumode quantum mechanics as well as the mathematical results needed to use these qumode operations to compile a given unitary transformation.

\subsection{A hybrid \qq~device}
We first produce a mathematical description of a \qumode, then describe operations which can be performed on the \qumode~and qubit (\cref{fig:device-picture}).
Qumodes store bosonic states. Bosons are are already commonplace in quantum computing experiments: photons (energy quanta of the electromagnetic field) are used in photonic chips, cavity QED, and hybrid circuit QED, while phonons (quanta of mechanical vibrations) are used to couple ion-spin qubits in ion traps. However, we are interested in using the \qumode~as an explicit computational resource, rather than as a conduit for entangling operations or source of noise.

\subsubsection{Representing the qumode}
There are two different bases that are commonly used to describe the state of the qumode: 
\begin{enumerate}
    \item \textit{Phase-space representation}, where operators are written in terms of position ($\hat{x}$) and momentum ($\hat{p}$) operators
    \item \textit{Fock-space representation}, where operators are written in terms of \qumode~creation ($a^\dagger$) and annihilation ($a$) operators.
\end{enumerate}

In the phase-space representation, the computational basis corresponds to the strength of the electric field in the case of photons (or equivalently the position of a mechanical oscillator for vibrational systems). We refer to this with the operator $\hat{x}$ and we have that for any $x\in \mathbb{R}$, $\hat{x} \ket{x} = x \ket{x}$. We also use the corresponding operator for momentum $\hat{p}$.
This describes the magnetic field for a photonic system.  In practice, cutoffs are imposed on the values of the field and further discretization error on the gates and the outputs prevents arbitrary precision readout; however, for simplicity we ignore the latter issue in order to provide a simpler if less realistic computational model and ignore the issue that even when cutoffs are imposed the vector space does not strictly form a Hilbert space without also including spatial discretization.

In the Fock-space representation, we track the number of bosons (number of photons or the energy level of the harmonic oscillator for the vibrational case) in the computational basis.  In this representation the computational basis is defined to be an eigenvector of the boson number operator $\hat{n} \ket{n} = n \ket{n}$, where $\hat{n} = a^\dagger a$ is the number operator and $a$ and $a^\dagger$ add and remove a boson from the system, respectively.  Formally this spectrum is countably infinite, but after truncation it forms a finite-dimensional Hilbert space and thus can be thought of as a qudit.  For example, assuming a cutoff $\Lambda = 3$ on the boson number $\hat n$
\begin{align}
    P_3 a^\dagger P_3 = \begin{bmatrix}
    0 & 0 & 0 & 0 \\
    1 & 0 & 0 & 0 \\
    0 & \sqrt{2} & 0 & 0 \\
    0 & 0 & \sqrt{3} & 0 
    \end{bmatrix} \qquad P_3 a P_3 = \begin{bmatrix}
    0 & 1 & 0 & 0\\
    0 & 0 & \sqrt{2} & 0 \\
    0 & 0 & 0 & \sqrt{3} \\
    0 & 0 & 0 & 0
    \end{bmatrix}.
\end{align}
Here $P_\Lambda$ is the projector onto the subspace of the cavity containing at most $\Lambda$ photons
\begin{align}
    P_\Lambda : P_\Lambda \ket{n} = \begin{cases}
    \ket{n} & n \leq \Lambda \\
    0 & \textrm{otherwise}
    \end{cases}.
\end{align}

\change{Observe that this Hilbert space has dimension $\Gamma + 1$ corresponding to the Fock states from $\ket
0$ to $\ket{\Gamma}$.}  Truncation of the Hilbert space is required to provide error bounds, otherwise the remainder terms become undefined.
Provided that an appropriate cutoff is picked for the system, the discrepancies between the truncated and untruncated systems will often be negligible. For notational clarity, we assume a cutoff of $\Lambda$ for all further equations and assume the annihilation and creation operators implicitly have the projectors $P_\Lambda$.

To incorporate the qubit's state, we 
take the tensor product of the qubit and \qumode~Hilbert spaces, so that the state space is $ \mathcal{H}_2 \kron \mathcal{H}_{\Lambda + 1}$. This means that, for example, a computational basis state will be of the form $\ket{q} \kron \ket{m} $ where the state $\ket{q}$ here can be thought of as the union of the qubits in the system, and $\ket{m}$ represents a qumode state 
where the system is either in position $x=m$ for the phase-space encoding or has $m$ photons if the Fock-space encoding is used.

\subsubsection{Representing operations and measurements via block encodings}
Depending on whether the qumode is described in Fock or phase space, there exist two pairs of complementary \qumode~operators. In Fock space, the operators are creation and annihilation operators $a^\dagger, a$; in phase-space, the operators are position and momentum $\hat{x}, \hat{p}$.

Our techniques can be applied in both spaces because
Fock and phase-space operators have the equivalencies
\begin{align}
    \hat{x} = \frac{1}{2} (a + a^\dagger) \qquad &\Leftrightarrow \qquad a = \hat{x} + i \hat{p} \label{eq:AnnihilationOperatortoPhaseSpace},\\
    \hat{p} = -\frac{i}{2} (a - a^\dagger) \qquad &\Leftrightarrow \qquad a^\dagger = \hat{x} - i \hat{p}, \label{eq:CreationOperatortoPhaseSpace}
\end{align}
and commutation relations
\begin{align}
    [\hat{x}, \hat{p}] &= \frac{i}{2} \label{eq:Commutatorxp}, \\
    [a, a^\dagger] &= 1 \label{eq:CommutatorCreationAnnihilation}.
\end{align}

To illustrate the \qq~ISA, we define three types of operations with examples:
\begin{enumerate}
    \item \textbf{Qubit-exclusive}: these include typical qubit gates like the \textit{Pauli} ($X, Y, Z$) gates, \change{ADDED below:}
    \begin{align}
        X = \begin{bmatrix}
            0 & 1 \\
            1 & 0
        \end{bmatrix}, \qquad Y = \begin{bmatrix}
            0 & -i \\
            i & 0
        \end{bmatrix}, \qquad \begin{bmatrix}
            1 & 0 \\
            0 & -1
        \end{bmatrix},
    \end{align}
    and \textit{Hadamard} ($H$) and \textit{phase} ($S$) gates
     \begin{align}
        H = \frac{1}{\sqrt{2}} \begin{bmatrix}
            1 & 1 \\
            1 & -1
        \end{bmatrix}, \qquad S = \begin{bmatrix}
            1 & 0 \\
            0 & i
        \end{bmatrix}.
    \end{align}
    \item \textbf{Qumode-exclusive}: we assume that linear optical operations (which are at most quadratic in the field operators) can be performed on the \qumode. This includes the \textit{displacement} operations $e^{ \alpha a^\dagger + \alpha^* a }$ for $\alpha \in \mathbb{C}$, \textit{phase delays} (or phase-space rotations) $e^{-i\alpha  a^\dagger a }$, and \textit{squeezing} operations $e^{\alpha (a^\dagger)^2 -\alpha^{*}a^2}$. Qumodes can also be entangled with other qumodes via \textit{beamsplitter} operations $e^{(\alpha a^\dagger b - \alpha^{*}a b^{\dagger})}$ where $b$ is the creation operator acting on a different \qumode. In this work, we primarily focus on the single-\qumode~case.  The multi-mode case with beamsplitters is discussed in \cite{C2QA_ISA,C2QA_LGT} but without the rigorous convergence bounds provided in the present work.
    \item \textbf{Qumode-qubit entangling}: there are several entangling operations between individual \qumodes~and qubits that widely appear in the circuit QED literature. Two common operations we consider are the \textit{conditional displacement} operation \cite{eickbusch2021fast} $e^{-i\sigma^z \otimes (\alpha a^\dagger + \alpha a)}$ and the \textit{Selective Number-dependent Arbitrary Phase} (SNAP) gate \cite{SNAP-PhysRevLett.115.137002} $e^{-i  \sigma^z \otimes \sum_n\alpha_n \hat P_n }$, where $\hat P_n=|n\rangle\langle n|$ is the projector onto the $n$th Fock state.
\end{enumerate}
Note that, for clarity, we use uppercase letters for qubit-exclusive gates and lowercase letters for \qumode-exclusive or \qumode-qubit gates. For example, we use $S, X, H$ for qubit-exclusive gates and $\sigma^i$ notation for Paulis in hybrid gates.

Our compilation strategy describes operations in terms of qubit-exclusive gates and the $\mathcal{S}_1$ gate, a primitive \qumode-qubit entangling gate (\cref{defn:SX}). $\mathcal{S}_1$ is a useful ``block-encoding" primitive to compose complex gates because it embeds a first-order Fock operator in the off-diagonal blocks. ``Block-encoded" matrices refer to how ``blocks"/submatrices of a larger matrix can be expressed as an existing matrix. Block encodings are frequently used in quantum algorithm design \cite{martyn2021grand, camps_approximate_2020}.

\begin{define}[$\mathcal{S}_1$ primitive gate]\label{defn:SX}
For any $t>0$ and any positive integer cutoff $\Lambda$, we define $\mathcal{S}_1$ to be the unitary acting on the Hilbert space $\mathcal{H}_2 \otimes \mathcal{H}_{\Lambda + 1}$ that has the following representation as a block matrix
\begin{align}
    \mathcal{S}_1 = \exp \left(it \begin{bmatrix}
    0 & a^{\dagger} \\
    a & 0
    \end{bmatrix} \right).
\end{align}
Note that $\mathcal{S}_1$ can itself be decomposed into conditional displacements (see \cref{obtaining_s1}) or implemented directly via OCT \cite{PhysRevX.4.041010,PhysRevA.91.043846,Rosenblum2018,Rosenblum2023}. Also note that the block encoding can also be expressed as the sum of qumode-qubit tensor products
\begin{align}\label{eq:JC}
    \begin{bmatrix}
    0 &  a^\dagger  \\
     a & 0
    \end{bmatrix} = \ket{0}\bra{1} \kron a^\dagger  + \ket{1}\bra{0} \kron  a.
\end{align}
\end{define}

We also consider the broader class of block-encoded Hamiltonians:
\begin{define}[Block encodings]
For a qumode operator $A$ acting upon $\mathcal{H}_{\Lambda + 1}$, we denote the joint block-encoded Hamiltonian to be
    \begin{align}\label{eq:blockencodedmatrix}
        \mathcal{B}_A = \begin{bmatrix}
            0 & A \\
            A^\dagger & 0
        \end{bmatrix},
    \end{align}
so that the subscript is the upper right block and the lower left block is the transpose and complex conjugate to preserve Hermiticity. I.e., $\mathcal{B}_A$ is Hermitian for any $A$, and thus is a suitable Hamiltonian. Note that $\mathcal{B}_A$ can describe a \qq~Hamiltonian acting on $\mathcal{H}_2 \otimes \mathcal{H}_{\Lambda + 1}$. Furthermore, if $A$ Hermitian, $\exp it \mathcal{B}_A = \exp{it\sigma^x A}$, otherwise it corresponds to $\exp{it( \ket{0}\bra{1} \kron A  + \ket{1}\bra{0} \kron A^\dagger)}$.
\end{define}

In this block-encoding notation, observe that $\mathcal{S}_1 = \exp it  \mathcal{B}_{a^\dagger}$. 
Throughout this work, we consider increasingly exotic $A$ matrices that can be created via polynomials of qumode operators.

Finally, to simplify notation, we oftentimes write tensor products implicitly as follows
\begin{align}
    \sigma^i \otimes M = \sigma^i M.
\end{align}
This notation is used extensively to abbreviate the action of $\sigma^i$  on $\mathcal{H}_2$ and $M$ (which is  comprised of qumode operators) on $\mathcal{H}_{\Lambda + 1}$.

We further assume that the qubit can be measured directly, but the \qumode~can only be measured by entangling it with a qubit and reading out the state of the qubit to obtain a single classical bit of information about the qumode state \cite{C2QA_ISA}.

\subsection{Matrix product formulas}
\begin{figure}
    \centering
    \includegraphics[width=\textwidth]{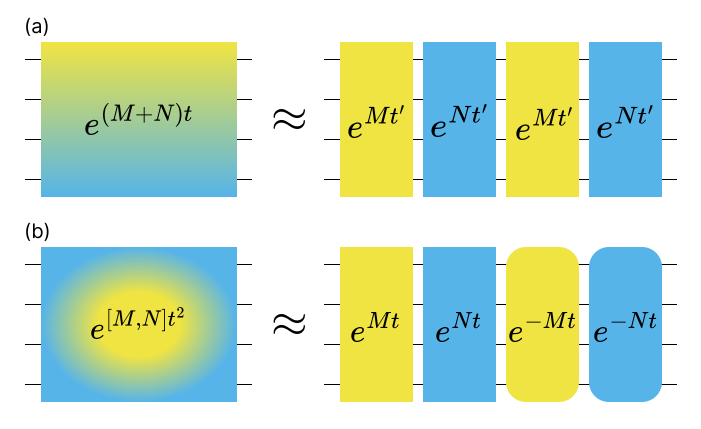}
    \caption{A visualization of the Trotter and BCH formulas. (a) A Trotterized formula that approximates $e^{(M + N) t}$ via $e^{M t}, e^{N t}$ evaluated at $t'$. (b) A BCH product formula which approximates $e^{[M, N] t^2}$ via $e^{M t}, e^{N t}$ evaluated at $t, -t$.}
    \label{fig:product-formula-visualization}
\end{figure}

Matrix product formulas describe the behavior of products of matrix exponentials (namely $e^A e^B$). These formulas are well-known in Hamiltonian simulation~\cite{childs2021theory,berry2007efficient,su2021fault} and are used to approximate a discretized version of the time evolution operator $e^{-iHt}$ using the Trotter formula. Thus, they often have rigorous error bounds that describe how $e^{-iHt}$ can be approximated given constituent $e^{-iH_jt}$. In our setting, $e^{-iHt}$ will be the operation we seek to implement using the $\mathcal{S}_1$ gate. 

As stated above, we use two product formulas: the BCH formula and the Trotter formula (\cref{fig:product-formula-visualization}). The BCH formula is used to create a commutator (or anticommutator) of operators. The Trotter formula is used to add these commutators and anticommutators together. We introduce the informal theorems below:
\begin{theorem}[Informal Trotter theorem from \cite{berry2007efficient}]
    Suppose we may implement $e^{M \lambda}$ and $e^{N \lambda}$ for arbitrary $\lambda \in \mathbb{R}$ and anti-Hermitian $M, N$. Then, a $p^\text{th}$ order Trotter formula has the error scaling
    \begin{align}
        \trotter_{2p}(M\lambda, N \lambda) & = e^{(A+B)\lambda} + \mathcal{O}((\norm{M + N} \lambda)^{2p + 1}),
    \end{align}
    requiring no more than $4 \cdot 5^{ p -1}$ exponentials.
\end{theorem}

\begin{theorem}[Informal BCH theorem from \cite{Childs_2013}]\label{thm:BCH}
    Suppose we can implement the operators $e^{M \lambda}, e^{N \lambda}$ for $\lambda \in \mathbb{R}$ and anti-Hermitian $M, N$. Then, a BCH formula of order $p$ has the error scaling
    \begin{align}
        \bch_p(M\lambda, N \lambda) &= e^{[M, N] \lambda^2} + \mathcal{O}((\max (\norm{M}, \norm{N}) \lambda)^{2p + 1}), 
    \end{align}
    requiring no more than $8 \cdot 6^{ p -1}$ exponentials.
\end{theorem}

These product formulas are defined recursively and are comprised of sequences of $e^{M \lambda}, e^{N \lambda}$ gates evaluated at varying values of $\lambda$. While the sequences may be long, as there is an exponential dependence on $p$, it is important to note that these error bounds are known to be loose, so experimental implementations may benefit from improved accuracy. 

Finally, note in our notation these formulas take the logarithm matrix and output sequences of matrix exponentials which approximate a desired exponential. Additionally, note that Hamiltonians can be `conjugated', i.e. we can transform a Hamiltonian from $H$ to $U H U^\dagger$ for any unitary $U$ via:
\begin{align}
    U \exp (-i t H) U^\dagger = \exp (-it U H U^\dagger) 
\end{align}
Conjugation, Trotter, and BCH give us immense flexibility for the Hamiltonian manipulations which can be achieved (\cref{tab:all-formulas}).

\section{Producing Anticommutators and Exponential Products}
In this section, we formalize our technique to build polynomials of block encodings in both Fock and phase space. This is notable because it shows how Hermiticity is not a requirement for our compilation scheme: whereas phase-space operators $\hat{x}$ and $\hat{p}$ are Hermitian, Fock operators $a, a^\dagger$ are decidedly non-Hermitian.

We first show how our hybrid \qq~architecture allows for the synthesis of anticommutators of Hermitian operators and, by proxy, matrix products of phase-space operators. We then use similar techniques with non-Hermitian operators such as Fock-space operators to manipulate block encodings of matrices. Finally, we contextualize these methods with asymptotic error bounds, providing theoretical analyses of our proposed techniques.

\begin{table}
    \small
    \centering
    \begin{tabularx}{\textwidth}{ccll}
        \hline
        \textbf{Formula} & \textbf{Target} &\textbf{Preconditions} & \textbf{Reference} \\
        \hline
        $\textrm{BCH}(it \sigma^i B , it \sigma^i A )$ & $\exp( t^2 [A, B]) \identity$ &$A, B$ Hermitian &\cref{thm:BCH}  \\
        $\textrm{BCH}(it \sigma^j A , it \sigma^k B ) $ & $\exp(- i t^2 \sigma^i \{ A, B \} )$ &$A, B$ Hermitian &\cref{eq:sigmazterm} \\
                \hline
         $\textrm{BCH}(X \cdot  i t \mathcal{B}_B \cdot X, it \mathcal{B}_A )$ & $\exp\sigma^z (t^2 (AB - (AB)^\dagger) )$ & $[A,B]=0$ &\cref{eq:C2}\\
         $\textrm{BCH}(S \cdot it \mathcal{B}_A \cdot S^\dagger, X \cdot it \mathcal{B}_B \cdot  X) $ & $\exp(i t^2 \sigma^z (AB + (AB)^\dagger) ) $ & $[A,B]=0$ & \cref{eq:AC2}\\
         $\textrm{BCH}(S\cdot  it \mathcal{B}_A \cdot S^\dagger, X \cdot it \mathcal{B}_B \cdot X) $ & $\exp(2 i t^2 \sigma^z AB  )$ & \makecell[l]{$[A,B]=0$, \\$AB = (AB)^\dagger $} &\cref{eq:Product} \\
         $X \cdot \textrm{Trotter}\Big(t \sigma^y (AB - (AB)^\dagger) , $ & $\exp\Big(2 it \begin{bmatrix}
            0 & AB \\
            (AB)^\dagger & 0
        \end{bmatrix} \Big) $ & $[A,B]=0$ &\cref{thm:general-adder-error}\\
            \qquad$ it \sigma^x (AB + (AB)^\dagger) \Big) \cdot X$ & \\
         $\textrm{BCH}(S \cdot i t \mathcal{B}_{A} \cdot S^\dagger, X \cdot i t \mathcal{B}_{B} \cdot X)$ & $\exp\left(it \begin{bmatrix}
            2 AB & 0 \\
            0 & -BA - (BA)^\dagger
        \end{bmatrix}\right)$ & $AB = (AB)^\dagger$ &\cref{lem:multiplication-alg}
    \end{tabularx}
    \caption{
    Overview of techniques for synthesizing particular unitary transformations and the quantum gates needed. Our formulas allow the manipulations of broad classes of Hamiltonian block encodings, denoted $\mathcal{B}_A, \mathcal{B}_B$. Each row contains the formula used, the target to approximate, the preconditions, and a reference to the location of the precise statement of the performance of the method. Most formulas use `conjugated' hamiltonians ($U H U^\dagger$) which can be achieved via conjugation of the exponential. The formula provided denotes hybrid gates with $\sigma^i$ terms, and single-qubit gates are capitalized (e.g., $S, X, H$). The bounds on the number of gates depend on the accuracy required of the approximation and are given in the corresponding theorems. 
    }
    \label{tab:all-formulas}
\end{table}

\subsection{Intuition of polynomial building via anticommutators in phase space}
Our \qq~architecture uses BCH to natively implement anticommutator-like exponentials by exploiting the qubit's Hilbert space. In conjunction with Trotter, these anticommutators are used to build larger block encodings; we generate nontrivial transformations like beamsplitters and the Hong-Ou Mandel effect in~\cref{sec:applications}.

To begin, we recall the  \qq~commutators
\begin{align}
    [\hat{x}, \hat{p}] = \frac{i}{2} \qquad \sigma^{i} & =\epsilon_{ijk}\frac{i}{2}\left[\sigma^{j},\sigma^{k}\right],\label{eq:PauliCommutator}
\end{align}
where $\epsilon_{ijk}$ is the Levi-Civita symbol, and $i,j,k \in \{x,y,z\}$. We use these relations, as well as the Pauli product identity 
\begin{align}
\sigma^{i} & =\epsilon_{ijk}i\sigma^{j}\sigma^{k}\label{eq:PauliProduct}
\end{align}
to decompose the anticommutator $\sigma_{i} \left\{ A,B\right\} $ in terms of a \qumode-qubit commutator
\begin{align}
    [i \sigma^j A, i \sigma^k B] &= - \sigma^j \sigma^k AB + \sigma^k \sigma^j BA \nonumber \\
    &= i \epsilon_{ijk} \sigma^i AB - i \sigma^i BA \nonumber \\
    &= i \epsilon_{ijk} \sigma^i \{ A, B \} \label{eq:sigmazterm}.
\end{align}
Note this assumes $A, B \in \mathcal{H}_{\Lambda + 1}$ are Hermitian so that $\exp[i \sigma^j A]$ is unitary and commutes with $\sigma^{i}$, as is the case when $A, B$ are mode-only operators. Thus, by using a hybrid qubit-cavity operation of the form $\exp i A \sigma^j, \exp i B \sigma^k$, the BCH formula can convert commutators into anticommutators. 

Finally, because $\frac{1}{2}[A, B] + \frac{1}{2} \{A , B\} = AB$, we use the Trotter formula to produce
\begin{align}
    \trotter \left(\frac{1}{2} \sigma^i [A, B], \frac{1}{2} \sigma^i \{ A, B \} \right) \approx \exp (\sigma^i AB),
\end{align}
assuming we may implement $\exp \left(\frac{1}{2} \sigma^i [A, B]\right)$ via a traditional BCH formula.

\subsection{Polynomial building non-Hermitian block encodings in Fock space}
In this section, we show how to achieve $A^q$ for an arbitrary \qumode~operator $A$. We also show that our techniques work in a multi-\qumode~setting. This extends the prior techniques, which require $A$ to be Hermitian. When $A$ is a Fock space operator, this corresponds to realizing arbitrary powers of $a$ and $a^{\dagger}$, which is known to generate a universal set of operations on the \qumode. This is useful for the simulation of nonlinear materials, which naturally lead to terms that are polynomial in $a, a^\dagger$, as well as quantum signal processing~\cite{martyn2021grand}. 

Our method again uses the qubit coupling to induce a phase in the compound \qq~system, similar to the previous section. We begin with block encodings as described in Eq.~\ref{eq:blockencodedmatrix}. 
To manipulate the block encodings, we begin by recognizing that qubit-only operations can modify the exponential via ``conjugation,",i.e.
\begin{align}
    U e^A U^\dagger = e^{U A U^\dagger}.
\end{align}
Thus, given any block encoding, we can also create the auxiliaries
\begin{align}
    X \cdot \exp \left( it \mathcal{B}_A \right) \cdot X &=  \exp \left( it X \cdot \begin{bmatrix}
        0 & A \\
        A^\dagger & 0
    \end{bmatrix} \cdot X \right) = \exp it \begin{bmatrix}
        0 & A^\dagger \\
        A & 0
    \end{bmatrix} = \exp \left( it \mathcal{B}_{A^\dagger} \right), \\
    S \cdot \exp \left( it \mathcal{B}_A \right) \cdot S^\dagger &= \exp \left( it S \cdot \begin{bmatrix}
        0 & A \\
        A^\dagger & 0
    \end{bmatrix} \cdot S^\dagger \right) = \exp it \begin{bmatrix}
        0 & -i A \\
        i A^\dagger & 0
    \end{bmatrix} = \exp (it \mathcal{B}_{-iA}),
\end{align}
recalling that $S$ is a qubit phase gate. Applying BCH yields the commutators
\begin{align}
    \Big[X \cdot it \mathcal{B}_A \cdot X, it \mathcal{B}_A \Big] 
    &= - t^2 \left( \begin{bmatrix}
        (A^\dagger)^2 & 0 \\
        0 & A^2
    \end{bmatrix} - \begin{bmatrix}
        A^2 & 0 \\
        0 & (A^\dagger)^2
    \end{bmatrix} \right) \nonumber \\
    &= t^2 \sigma^z (A^2 - (A^\dagger)^2), \\
    \left[S \cdot it \mathcal{B}_A \cdot S^\dagger, X \cdot it \mathcal{B}_A \cdot X \right] 
    &= -t^2 \left( \begin{bmatrix}
        - i A^2 & 0 \\
        0 & i (A^\dagger)^2
    \end{bmatrix} - \begin{bmatrix}
        i (A^\dagger)^2 & 0 \\
        0 & - iA^2
    \end{bmatrix} \right) \nonumber \\
    &= i t^2 \sigma^z (A^2 + (A^\dagger)^2).
\end{align}
These commutators themselves can be conjugated. Recall that $H Z H = X$ and $SH Z HS^\dagger = Y$
\begin{align}
    SH \cdot \left[it \mathcal{B}_{A^\dagger}, it \mathcal{B}_A \right] \cdot HS^\dagger &=t^2 \sigma^y (A^2 - (A^\dagger)^2), \\
    H \cdot \left[it \mathcal{B}_{-iA} , it\mathcal{B}_{A^
    \dagger} \right] \cdot H &= i t^2 \sigma^x (A^2 + (A^\dagger)^2),
\end{align}
so that
\begin{align}
    i t^2 \sigma^x (A + (A^\dagger)^2) + t^2 \sigma^y (A^2 - (A^\dagger)^2) = 2 it^2 \begin{bmatrix}
        0 & (A^\dagger)^2 \\
        A^2 & 0
    \end{bmatrix}.
\end{align}
Thus, using only $\mathcal{B}_A(t)$ gates, we can approximate $\mathcal{B}_{A^2}$. 

We now lift this procedure to produce $\exp it \mathcal{B}_{AB}$ given $\exp it \mathcal{B}_A, \exp it \mathcal{B}_B$ for commuting $A, B \in \mathcal{H}_{\Lambda + 1}$. Observe the following commutators (whose exponentials we can implement via BCH)
\begin{align}
    \left[ i \tau \mathcal{B}_{B^\dagger}, i \tau \mathcal{B}_A \right]
    &= \tau^2 \begin{bmatrix}
        AB - (AB)^\dagger & 0 \\
        0 & (BA)^\dagger - BA
    \end{bmatrix}, \\
    \left[ i \tau \mathcal{B}_{-iA}, i \tau \mathcal{B}_{B^\dagger} \right]
    &= i\tau^2 \begin{bmatrix}
        AB + (AB)^\dagger & 0 \\
        0 & -BA - (BA)^\dagger
    \end{bmatrix}. \label{eq:also-alg2}
\end{align}
Provided that $[A, B] = 0$ and via conjugation,
the equality simplifies to
\begin{align}
    SH \cdot  \Big[ it \mathcal{B}_{B^\dagger}, it \mathcal{B}_{A}\Big] \cdot HS^\dagger &= 
    \tau^2 \sigma^y(AB - (AB)^\dagger)  \label{eq:C2},\\
    H \cdot \Big[ it \mathcal{B}_{-iA}  , it \mathcal{B}_{B^\dagger}  \Big] \cdot H &= 
    i\tau^2 \sigma^x (AB + (AB)^\dagger) \label{eq:AC2}.
\end{align}
Via Trotter, we can directly implement the sum:
\begin{align}
    \tau^2 (AB - (AB)^\dagger) \sigma^y  + i\tau^2 \sigma^x  (AB + (AB)^\dagger)   = 2 i \tau^2 \begin{bmatrix}
        0 & (AB)^\dagger \\
        AB & 0
    \end{bmatrix}\label{eq:Product}.
\end{align}
We select $\tau = \sqrt{\frac{t}{2}}$ to obtain the desired time and conjugate by $\sigma^x$ to produce the desired matrix. 

This procedure is described in \cref{alg:adder} and thus allows us to approximate $\exp i t \mathcal{B}_{AB}$. This process can be repeated iteratively, assuming $AB$ commutes with $B$; for example, if $A = B = a$, then this process can be used to produce higher powers $a^k, (a^\dagger)^k$.

Our formulas require $[A, B] = 0$ to build higher order polynomials. This requirement is tolerable, as we still may achieve a broad class of transformations, including homogeneous polynomials of $a$ or $a^\dagger$. 
Our formulas can also be extended to more general cases where the operators commute, e.g. when the synthesized unitary operates on two different \qumodes, as in the conditional beamsplitter (a gate that acts as a beamsplitter controlled on an ancillary qubit).

\cref{alg:adder} is an extension of the prior commutator approaches in phase space because the $\sigma^i = - \frac{i}{2} [\sigma^j, \sigma^k]$ relation is natively expressed in the algorithm; i.e., if we have $\mathcal{B}_{A} = \mathcal{B}_B =  \mathcal{B}_{\hat{x}} = \exp it \hat{x} \sigma^x$, the ``$\textrm{Left}$" term vanishes and the``$\textrm{Right}$" term is the commutator we would apply.

Finally, in \cref{alg:mult}, we demonstrate how to implement $AB$ if $AB = (AB)^\dagger$. This process places $AB$ in the upper left block, which is useful to exactly execute $\exp it AB$ by preparing the qubit state to $\ket{0}$, but it prevents the process from being executed recursively. Simplifying the commutator from \cref{eq:also-alg2} finds
\begin{align}
    \left[ i \tau \mathcal{B}_{-i A}, i \tau \mathcal{B}_{B^\dagger}  \right] = i \tau^2 \begin{bmatrix}
        2 AB & 0 \\
        0 & -BA - (BA)^\dagger
    \end{bmatrix}.
\end{align}

\subsection{Error analysis}
The prior description of our algorithm assumes errorless product formulas. However, the BCH and Trotter formulas  introduce errors which must be accounted for, especially when applying our algorithm recursively. In this section, we cite the error scaling of the general addition algorithm described in \cref{alg:adder} and the multiplication algorithm described in \cref{alg:mult}. The proofs and full results are included in \cref{apndx:error-analysis}.

\begin{algorithm}[t]
\caption{ADD($\mathcal{B}_A (t), \mathcal{B}_B(t), p_l, p_r, t$) }\label{alg:adder}
\begin{algorithmic}
\Require $[A, B] = 0$, $\norm{\mathcal{B}_A - \exp it \begin{bmatrix}
    0 & A \\
    A^\dagger & 0
\end{bmatrix} } \in \mathcal{O}((c_A t)^{p_A})$, $\norm{\mathcal{B}_B - \exp it \begin{bmatrix}
    0 & B \\
    B^\dagger & 0
\end{bmatrix} } \in \mathcal{O}((c_B t)^{p_B})$, $p_A, p_B \geq 1$, $t > 0$
\Ensure $\mathcal{B}_{AB}$ where $\norm{\mathcal{B}_{AB}(t) - \exp it \begin{bmatrix}
    0 & AB \\
    (AB)^\dagger & 0
\end{bmatrix}} \in \mathcal{O}((Ct)^{\min (p_A, p_B) / 2})$ for constant $C$
\State $q \coloneqq \max ( \ceil{\frac{1}{2}(\min(p_l, p_r) - 1)}, 1 )$
\State $s \coloneqq \max ( \ceil{\frac{1}{2} ( \min(p_l, p_r) - 1)}, 1 )$
\State $\tau \coloneqq \sqrt{t / 2}$
\State Left $\coloneqq \bch_{q, 1}(X \cdot i \tau \mathcal{B}_B \cdot X, i \tau \mathcal{B}_A)$
\State Right $\coloneqq \bch_{q, 1}(S\cdot i \tau  \mathcal{B}_A \cdot S^\dagger, X \cdot i \tau  \mathcal{B}_B \cdot X)$
\State Left' $\coloneqq SH \cdot \textrm{Left} \cdot HS^\dagger$
\State Right' $\coloneqq H \cdot \textrm{Right} \cdot H$ 
\State \Return $X \cdot \trotter_s(\textrm{Left'}, \textrm{Right'}) \cdot X$
\end{algorithmic}
\end{algorithm}

\begin{algorithm}[t]
\caption{MULT($\mathcal{B}_A (t), \mathcal{B}_B(t), p_l, p_r, t$) }\label{alg:mult}
\begin{algorithmic}
\Require $AB = (AB)^\dagger$, $p_l, p_r \geq 1$, $t > 0$
\Ensure An upper-left block encoding $\mathcal{M}_{AB}$ where $\norm{\mathcal{M}_{AB}(t) - \exp it \begin{bmatrix}
    AB & 0 \\
    0 & \frac{1}{2} ( - BA - (BA)^\dagger)
\end{bmatrix}} \in \mathcal{O}( (C^2 t)^{\min (p_A, p_B) / 2} )$ for constant $C$
\State $q \coloneqq \max ( \ceil{\frac{1}{2}(\min(p_l, p_r) - 1)}, 1 )$
\State $\tau \coloneqq \sqrt{t / 2}$
\State \Return $\textrm{BCH}_q(S \cdot i \tau \mathcal{B}_{A} \cdot S^\dagger, X \cdot i \tau \mathcal{B}_{B} \cdot X) $
\end{algorithmic}
\end{algorithm}

\begin{restatable}{theorem}{algproduct}\label{thm:general-adder-error}
    Suppose we have approximations $\widetilde{\mathcal{B}}_A(t), \widetilde{\mathcal{B}}_B(t)$ with the error scaling
    \begin{align}
        \norm{\exp it \widetilde{\mathcal{B}}_A - \exp it \mathcal{B}_A} &\in \mathcal{O}( (ct)^{p_A} ), \\
        \norm{
        \exp it\widetilde{\mathcal{B}}_B - \exp it\mathcal{B}_B} &\in \mathcal{O} ((ct)^{p_B} ), 
    \end{align}
    for some constant $c$ and order $p_A, p_B \geq 1$ where $[A,B]=0$. Then, the application of \cref{alg:adder} will yield the scaling
    \begin{align}
        \norm{{\rm{ADD}} ( it \widetilde{\mathcal{B}}_A, it \widetilde{\mathcal{B}}_B) - \exp it \mathcal{B}_{AB}} \in \mathcal{O} \left((C_{TOTAL} t)^{\min(p_A, p_B) / 2}\right),
    \end{align}
    with $C_{TOTAL} = \max( \norm{AB}, \norm{BA}, C_{BCH}^2)$ and $C_{BCH} = \max( \norm{A}, \norm{B}, c)$,
    using no more than $ 1.07 \cdot 30^q $ exponentials, where $q = \max (\ceil{\frac{\min(p_1, p_2) - 1}{2}}, 1) $. 
\end{restatable}

\begin{restatable}{theorem}{algmult}\label{lem:multiplication-alg}
    Suppose we have approximate block encodings $\widetilde{\mathcal{B}}_A, \widetilde{\mathcal{B}}_B$  with the error
    \begin{align}
        \norm{\exp it \widetilde{\mathcal{B}}_A - \exp it \begin{bmatrix}
            0 & A \\
            A^\dagger & 0
        \end{bmatrix}} &\in \mathcal{O}( (ct)^{p_A}), \\
        \norm{\exp it \widetilde{\mathcal{B}}_B - \exp it \begin{bmatrix}
            0 & B \\
            B^\dagger & 0
        \end{bmatrix}} &\in \mathcal{O} ((ct)^{p_B}),
    \end{align}
    for constant $c$ and $p_A, p_B \geq 1$ where $AB = (AB)^\dagger$.
    Then, \cref{alg:mult} has the error
    \begin{align}
        \norm{ {\rm{MULT}}( it \widetilde{\mathcal{B}}_A,   it \widetilde{\mathcal{B}}_B) - \exp it \begin{bmatrix}
            AB & 0 \\
            0 & \frac{1}{2}( - BA - (BA)^\dagger)
        \end{bmatrix}} \in \mathcal{O}\left((C^2 t)^{\min (p_A, p_B) / 2}\right),
    \end{align}
    with $C = \max(\norm{A}, \norm{B}, c)$, using no more than $8 \cdot 6^{q - 1}$ exponentials where $q = \max ( \ceil{\frac{\min(p_A, p_B) - 1}{2}}, 1) $.
\end{restatable}

While the asymptotic error analysis suggests that the cost of this method is onerous, we note that the product formulas often have overly pessimistic error scaling and operation counts \cite{zhao2021hamiltonian}. In the applications below, we provide numerical simulations which suggest our technique is more readily implementable than theory suggests.

\section{Applications}\label{sec:applications}
In this section, we show how our technique is a powerful tool for analytically realizing desired operations. This technique succeeds both for Hamiltonian simulation problems and general control problems. In particular, we show how the aforementioned physical intuition for a desired transformation is often sufficient to produce an approach to create desired operations.

We include applications in both phase and Fock space:
\begin{enumerate}
    \item \textbf{Phase-space} techniques are demonstrated to be useful in the case where displacements ($e^{(\alpha a^\dag + \alpha^* a)} = e^{i \alpha \hat{x}}$ for $\alpha$ real or $e^{\alpha \hat{p}}$ for $\alpha$ imaginary) are the only experimentally available gates. We 
    produce the controlled parity operator
    $e^{it \sigma^z a^{\dagger}a}$ (\cref{apps-nondestructive-meas}); the beamsplitter $e^{-it \sigma^z (a^{\dagger}b + ab^{\dagger})}$ (\cref{subsec:Conditional-(Controlled-Phase)-Beam-Splitter}); gates for two encodings of universal control of the restricted $\text{span}\{\left|0\right>,\left|1\right>\}$ Hilbert space (\cref{subsec:Universal-Control}); and gates for simulation of Fermi-Hubbard lattice dynamics using the two lowest Fock states of the cavity (\cref{sec:Fermi-Hubbard-Lattice-Dynamics}), including same-site, hopping, controlled-beamsplitter (\cref{subsec:controlled-phase}), and FSWAP gates.
    \item \textbf{Fock space} techniques are shown to be useful assuming compilation to $\mathcal{S}_1$ (\cref{defn:SX}) 
    and single-qubit operations to produce polynomials of annihilation and creation operators, (namely $a^p {a^\dagger}^q$ for integer $p, q$). We demonstrate how polynomials of these operators can be used in Hamiltonian simulation (e.g. with $\chi^{(3)}$ nonlinear materials, \cref{apndx:error-jc}) and state preparation (\cref{subsec:state-prep}).
\end{enumerate}

\subsection{Nonlinear Hamiltonian simulation}
As a simple application, let us consider the case of simulating a $\chi^{(3)}$ nonlinear material. These interactions commonly occur in nonlinear optics and appear when the index of refraction for a material varies linearly with the intensity of the electromagnetic field.  Such interactions can be modeled for a single \qumode~using the expression
\begin{equation}\label{eq:desired-ham}
    H = \omega a^\dagger a + \frac{\kappa}{2}(a^\dagger)^2 a^2.
\end{equation}
Our goal here is to examine the cost of a simulation of such a Hamiltonian in our model for time $t$ and error tolerance $\epsilon$ and to determine the parameter regimes within which a hybrid simulation using our techniques could provide an advantage with respect to a conventional qubit-based simulation of the Hamiltonian. 

Each of the terms can be approximated using formulas from \cref{tab:all-formulas}. The $\omega a^\dagger a$ term requires an embedding of Hermitian $a^\dagger a$, so that
\begin{align}
    \textrm{BCH} \left(S \cdot  i \tau_1 \mathcal{B}_{a^\dagger} \cdot S^\dagger, X \cdot  i \tau_1 \mathcal{B}_{a} \cdot X \right) = 
    2 i \tau_1^2 \begin{bmatrix}
         a^\dagger   a  & 0 \\
        0 & - a a^\dagger  
    \end{bmatrix}.
\end{align}
The second order term is treated in the same way, noting that $\mathcal{B}_{(a^\dagger)^2}$ can be produced via \cref{tab:all-formulas}, so that
\begin{align}
    \textrm{BCH} \left(S \cdot  i \tau_1 \mathcal{B}_{(a^\dagger)^2} \cdot S^\dagger, X \cdot  i \tau_1 \mathcal{B}_{a^2} \cdot X \right) = 
    2 i \tau_2^2 \begin{bmatrix}
        (a^\dagger)^2 a^2 & 0 \\
        0 & - a^2 (a^\dagger)^2
    \end{bmatrix}.
\end{align}
Thus, via the BCH formula, we can block-encode the two Hamiltonian terms. Trotterizing allows us to block-encode the entire Hamiltonian into the upper-left quadrant. Thus, by setting the qubit to $\ket{0}$, we can approximate the Hamiltonian. The error scaling is as follows and is proven in \cref{apndx:error-jc}:

\begin{restatable*}[Generating non-linear Hamiltonians]{theorem}{resultjc}\label{app:jaynes-cummings}
Let $H$ be the following non-linear Hamiltonian:
 \begin{align}
     H = \omega a^\dagger a + \frac{\kappa}{2}(a^\dagger)^2 a^2,
 \end{align}
(i.e. a Hamiltonian with a Kerr non-linearity). Let $t $ be the evolution time and $\epsilon$ be the target error tolerance.
For any positive integer $q$ we can approximate an exponential of the block-encoded Hamiltonian with error at most $\epsilon$ in the operator norm using $r e^{\mathcal{O}(q)}$ $\mathcal{S}_1$ operations where $r \in \Omega\left( \frac{(\Lambda^{4} t)^{1 + 1 / (q - \frac{3}{4})} }{\epsilon^{1 / (q - \frac{3}{4})}} \right)$. 
\end{restatable*}

This shows that we can perform a simulation of the dynamics within error $\epsilon$ using a number of operations within our instruction set that scales near-linearly with the evolution time and subpolynomially with $\epsilon$.  Further, this approach requires no ancillary memory and can be done with a single \qumode~and a qubit. In contrast, a qubit-only device would require a polylogarithmic number of qubits in $\Lambda$.

It is worth noting that in this case the ancillary qubit is not being used directly in the model.  Instead it is being used to control the dynamics and generate the appropriate nonlinear interaction between the photons present in the model.

\subsection{Nondestructive measurement of the qumode}\label{apps-nondestructive-meas}

We now demonstrate how the approach can extend beyond problems in Hamiltonian simulation. We begin with an example of the technique for control: In particular, we seek to perform a nondestructive measurement of the qumode in which we project the information into the qubit~\cite{fockreadout_wang_2020, Fockshotresolved_curtis_2021}.

To construct such a nondestructive measurement, we seek to implement  $e^{i t \hat{n} \sigma^z}$ where $\hat{n} = a^\dagger a$ is the number operator. If we could implement this gate for arbitrary $t$, we could perform phase estimation on the qubit to nondestructively project the qumode into a fixed number of bosons. This could be done by setting $t$ sufficiently small so that $t \Lambda \leq 2 \pi$ is calculable with phase estimation. Alternatively, for $t = \pi$, this operation checks the parity of the qumode and applies an RZ gate for odd parities. We employ the instruction set in the phase-space representation to synthesize the infinitesimal conditional rotation gate
\begin{equation}
U_{\text{rot},k}=e^{i\lambda^{2}\hat{n}\sigma^{k}}
\end{equation}
for $k=x,y,z$. We rewrite $\hat{n}$ in terms of the phase-space operators by recognizing:
\begin{align}
\hat{n} & =\hat{a}^{\dagger}\hat{a}\\
 & =\hat{x}^{2}+\hat{p}^{2}-\frac{1}{2}\label{eq:NumbertoPhaseSpace}.
\end{align}
Applying Eqs.~\ref{eq:AnnihilationOperatortoPhaseSpace},~\ref{eq:CreationOperatortoPhaseSpace}, and~\ref{eq:NumbertoPhaseSpace}
yields
\begin{equation}
i\lambda^{2}\hat{n}\sigma^{k}=i\left(\hat{x}^{2}+\hat{p}^{2}-\frac{1}{2}\right)\sigma^{k}\lambda^{2},
\end{equation}
such that the gate is expressed via the Trotter decomposition as the product of $\exp\left(\left[A_{1},B_{1}\right]\right)=\exp\left(i\lambda^{2}\hat{x}^{2}\sigma^{k}\right)$,
$\exp\left(\left[A_{2},B_{2}\right]\right)=\exp\left(i\lambda^{2}\hat{p}^{2}\sigma^{k}\right)$,
and conditional displacement $\exp\left(-i\lambda^{2}\sigma^{k}/2\right)$.
Given the Pauli commutator relation, the
first commutator is 
\begin{align}
\left[A_{1},B_{1}\right] & =i\hat{x}^{2}\sigma^{k}\\
 & =i\hat{x}^{2}\left(-\frac{i}{2}\left[\sigma^{i},\sigma^{j}\right]\right)\\
 & =\left[\frac{1}{\sqrt{2}}\hat{x}\sigma^{i},\frac{1}{\sqrt{2}}\hat{x}\sigma^{j}\right],
\end{align}
and the second commutator is 
\begin{align}
\left[A_{2},B_{2}\right] & =i\hat{p}^{2}\sigma^{k}\\
 & =i\hat{p}^{2}\left(-\frac{i}{2}\left[\sigma^{i},\sigma^{j}\right]\right)\\
 & =\left[\frac{1}{\sqrt{2}}\hat{p}\sigma^{i},\frac{1}{\sqrt{2}}\hat{p}\sigma^{j}\right],
\end{align}
such that both terms are amenable to BCH decomposition,
and the infinitesimal conditional rotation is composed with a gate-depth lower bound of nine. To perform an error analysis, we may directly apply the error scaling of BCH and Trotter to find:

\begin{restatable*}{theorem}{resultmeasurement}
Suppose we can implement $\text{e}^{ it \hat{x} \sigma^i}, \text{e}^{ it \hat{p} \sigma^i}$ without error. Then, we may approximate $\exp it \mathcal{B}_{\hat{x}^2 + \hat{p}^2}$ with arbitrary error scaling $p$ as
\begin{align}
    \norm{ \exp it \widetilde{\mathcal{B}}_{\hat{x}^2 + \hat{p}^2} - \exp it \mathcal{B}_{\hat{x}^2 + \hat{p}^2} } \in \mathcal{O}( (C t)^{p + 1/ 2}),
\end{align}
where $C = \max( \norm{ \hat{x}^2 + \hat{p}^2 }, \norm{\hat{x}}^2, \norm{\hat{p}}^2)$ and using no more than $4 \cdot 5^{ \frac{p}{2} - \frac{1}{4} }$ exponentials.
\end{restatable*}

We then provide numerics in \cref{fig:ConditionalPhase}. As expected, the wavefunction
initialized in the second excited state of the cavity and the ground
state of the associated qubit has an autocorrelation function that
oscillates with phase $\exp(2it)$. Dynamics are well reproduced with
$2000$ time steps for a final time of $20$ with
cutoff $\Lambda=14$. Note the units are arbitrary in the absence
of definition of the cavity frequency $\omega$, with the only units
defined by setting the reduced Planck constant to unity $\hbar=1\text{ arb. units}$.
The close agreement between the BCH-synthesized and exact gates is
supported by the error scaling after a single gate application computed for time step $t$, which features a power law scaling in agreement with the predicted error scaling
for both BCH and Trotter decompositions.\begin{figure}[!ht]
\begin{centering}
\includegraphics[width=0.5\columnwidth]{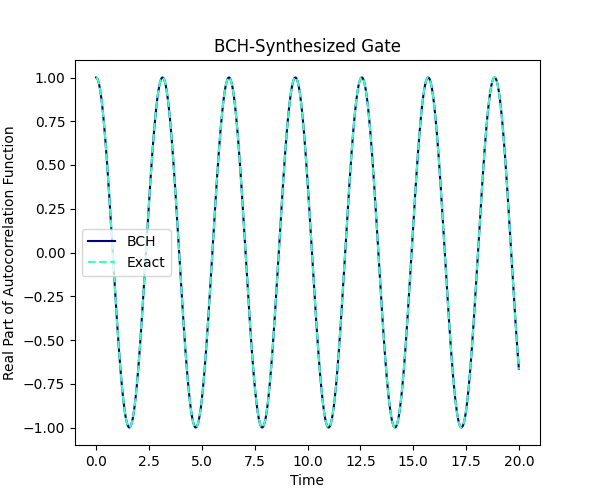}\includegraphics[width=0.5\columnwidth]{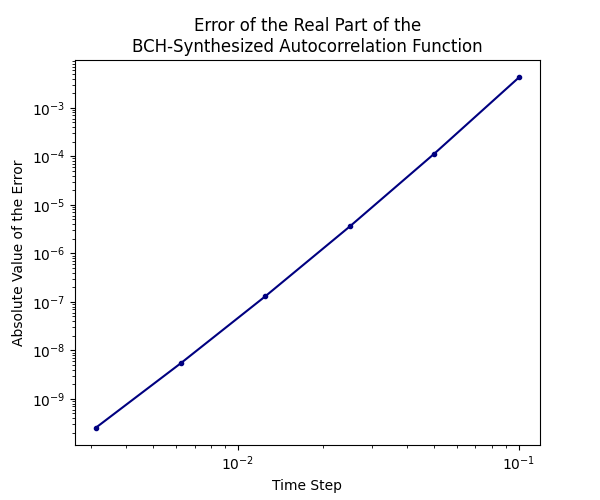}
\par\end{centering}
\caption{The following plots characterize the performance of using phase-space operators to synthesize $U_{\text{rot},k} = e^{i \hat{n} \sigma^z t}$. (a) The BCH-synthesized conditional rotation gate 
successfully approximates the exact dynamics for a wavefunction initialized
in the ground state of the qubit and the second excited state of
the cavity.  (b) Error of the real part of the autocorrelation
function $\text{Re}\left(\left<\psi(0)\middle|\psi(t)\right>\right)$ for the BCH-synthesized gate after a single time step of length $t=0.01$.\label{fig:ConditionalPhase}}
\end{figure}

We can obtain a similar decomposition with Fock-space operators. Observe that the $\textrm{MULT}$ subroutine applied to the $\exp it \mathcal{B}_a$, $\exp it \mathcal{B}_{a^\dagger}$ using \cref{lem:multiplication-alg} can also yield 
\begin{align}
	\norm{\textrm{MULT}( it \mathcal{B}_a, it \mathcal{B}_{a^\dagger}) - \exp i t \begin{bmatrix}
		a^\dagger a & 0 \\
		0 & - a a^\dagger
	\end{bmatrix} } \in \mathcal{O}((C^2 t)^{p + 1 / 2}),
\end{align}
when $\mathcal{B}_{a}$'s implementation is error-free. 
Note that $a a^\dagger = a^\dagger a + \identity$,  so the block encoding that is actually applied is
\begin{align}
\exp i t\left(\begin{bmatrix}
	a^\dagger a & 0 \\
	0 & - a^\dagger a 
\end{bmatrix} + \begin{bmatrix}
	0 & 0 \\
	0 & - \identity
\end{bmatrix}\right) = \exp it ( \sigma^z \hat{n}  -  (\identity - \sigma^z) \identity_{\Lambda + 1} ).
\end{align}
Thus, our Fock-space methods would also achieve the same transformation, albeit requiring a phase and RZ correction.

\subsection{State preparation from the vacuum}
Consider the case where we seek to prepare Fock state $\ket{k}_b$ on the qumode. On qubit devices, preparing integer states is trivial, assuming the qubits represent logic in a binary fashion. However, hybrid boson-qubit devices natively implement exponentials of the phase-space or Fock-space operators. Thus, preparing $\ket{k}_b$ directly can often be challenging. Existing work \cite{law_arbitrary_1996} enables the preparation of states, but are less granular in their control of the \qq~state. We present a method that only swaps the $\ket{1} \ket{0} \leftrightarrow \ket{0} \ket{k}$ states for integer $k$, otherwise leaving the initial state intact. 

To begin, we aim to implement $(a^\dagger)^k$ on the vacuum. It is sufficient to approximate
\begin{align}
    \mathcal{T}_k(t) \coloneqq \exp i t\begin{bmatrix}
        0 & (a^\dagger)^k \\
        a^k & 0
    \end{bmatrix},
\end{align}
which we call the `unprotected' state preparation operator. 
Selecting appropriate $t$ yields precisely the desired behavior, which gives the following result:
\begin{restatable*}[Unprotected state preparation]{theorem}{statepreptime}\label{thm:statepreptime}
    For $k \leq \Lambda$, we can take $t = (2n + 1) \frac{\pi}{2 \sqrt{k!}}$ for any $n \in \mathbb{N}$ so that
    \begin{align*}    
    \mathcal{T}_k(t) \ket{1} \kron \ket{0} = \ket{0} \kron \ket{k}. 
    \end{align*}
\end{restatable*}

\newcommand{\diag}{{\rm{diag}}}
\newcommand{\stateprep}{\mathcal{P}_k}
\newcommand{\stateprepapprox}{\widetilde{\mathcal{P}}_{k, p}}

While $\mathcal{T}_k (t)$ performs the desired transformation, it may incur unwanted side effects if the starting state is of the form $\ket{1} \kron \ket{b}$ for $b > 0$. We can use our same approach to produce the following operation:

\begin{restatable*}[Protected state preparation]{theorem}{resultstateprep}\label{lem:fock-prep-unitary}
Consider the Fock preparation unitary $\stateprep$ with the form
\begin{align*}
    \exp \left( i t \begin{bmatrix}
        0 & ( a^\dagger )^k \ketbra{0}{0} \\
        \ketbra{0}{0} ( a )^k & 0
    \end{bmatrix} \right).
\end{align*}
When $t = (2n + 1) \frac{\pi}{4 \sqrt{k!}} $, we have that $\stateprep$ performs our desired state preparation
\begin{align*}
    \exp \left( i t \begin{bmatrix}
        0 & ( a^\dagger )^k \ketbra{0}{0} \\
        \ketbra{0}{0} ( a )^k & 0
    \end{bmatrix} \right) \ket{1} \kron \ket{b}  = \begin{cases}
    \ket{0} \kron  \ket{k} & b = 0 \\
    \ket{1} \kron  \ket{b} & b \neq 0
    \end{cases}.
\end{align*}
We claim that we can approximate this unitary with $\stateprepapprox$ where
\begin{align*}
    \norm{\stateprepapprox - \stateprep} \in \mathcal{O}((\Lambda^{k/2} t)^p),
\end{align*}
using no more than $4 \cdot 5^{q -1}$ $\widetilde{\mathcal{T}}_{k,p}$ subroutines.

\end{restatable*}
The proofs of \cref{thm:statepreptime}, \cref{lem:fock-prep-unitary} are provided in \cref{state_prep_proof}. Though this subroutine appears expensive, numerical results suggest it is far more implementable than theory would suggest. In the following simulations, we apply the above technique but always use a second-order symmetrized BCH formula and second-order (symmetrized) Trotter formula. This amounts to 480 exponentials for the unprotected case and 960 exponentials for the protected case. The synthesized gates are provided in \cref{fig:unprotectedt2} and \cref{fig:protectedt2}.  

\change{As can be seen in \cref{fig:both-errors}, the protected gate has higher compilation errors than the unprotected gate. This is likely due to the fact that the protected formula requires an additional application of the Trotter product formula. This will further incur commutator error over the unprotected gate.}

\begin{figure}[!ht]
    \centering
    \includegraphics[scale=0.6]{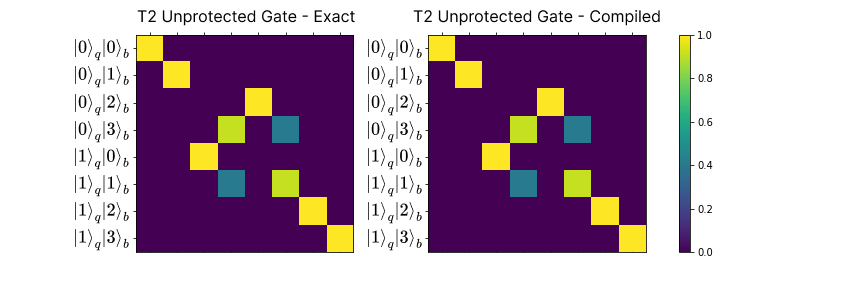}
    \caption{Unprotected $T_2$ Gate: plots visualize the unitary's elements. Brighter colors represent higher absolute values / magnitudes of the matrix entries. \change{\textit{Left}} is the exact form of the $T_2$ gate with selected $t$ and \change{\textit{right}} is the BCH-synthesized form. The state $\ket{j}_q \ket{k}_m$ has index $j \cdot \Lambda + k$ where $\Lambda$ is the cutoff. By observation, the analytically realized form is accurate.}
    \label{fig:unprotectedt2}
\end{figure}
\begin{figure}[!ht]
    \centering
    \includegraphics[scale=0.6]{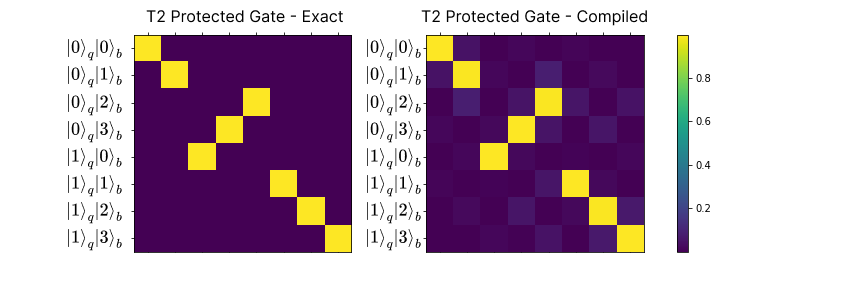}
    \caption{Protected $T_2$ Gate: \change{\textit{left}} is the exact form of the protected $T_2$ gate with selected $t$ and \change{\textit{right}} is the BCH-synthesized form.   The state $\ket{j}_q \ket{k}_m$ has index $j \cdot \Lambda + k$ where $\Lambda$ is the cutoff. By observation, the analytically realized form is still accurate, albeit with more incurred Trotter error. 
    }
    \label{fig:protectedt2}
\end{figure}

\begin{figure}
    \centering
    \includegraphics[width=0.8\linewidth]{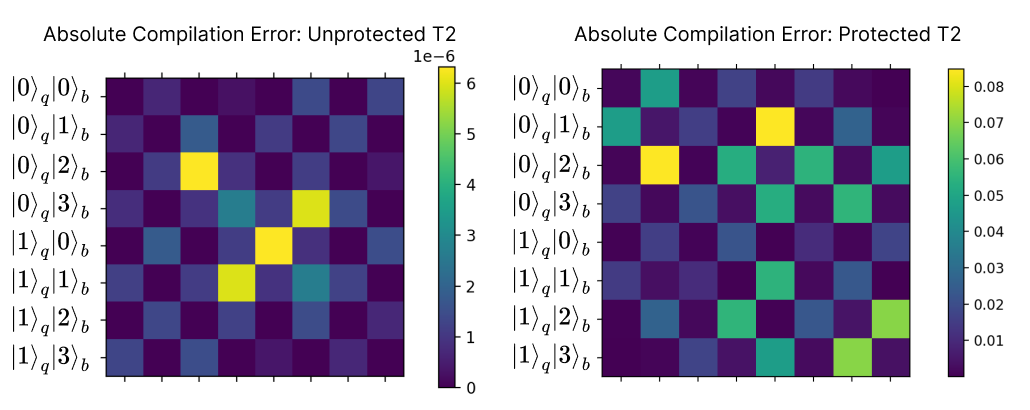}
    \caption{\change{Absolute element-wise error when compiling $T_2$ gate via scheme: \textit{left} is when the compilation target is the unprotected gate, while \textit{right} is when the target is the protected gate. The protected gate has higher synthesis errors.}}
    \label{fig:both-errors}
\end{figure}

We also analyze the error scaling as the order of the BCH formulas used increases. \cref{fig:protectedt2-error} describes the error-resource tradeoff as the Trotter step within each BCH formula increases. Observe that the error decays as higher order formulas are used or Trotter step size is reduced, as expected. While modest depths have relatively high infidelity, in theory our compilations can achieve arbitrary accuracy at the level of both entire gates and matrix individual elements. %
\begin{figure}[!ht]
    \centering
    \mbox{
    \subfigure[]{\includegraphics[scale=0.42]{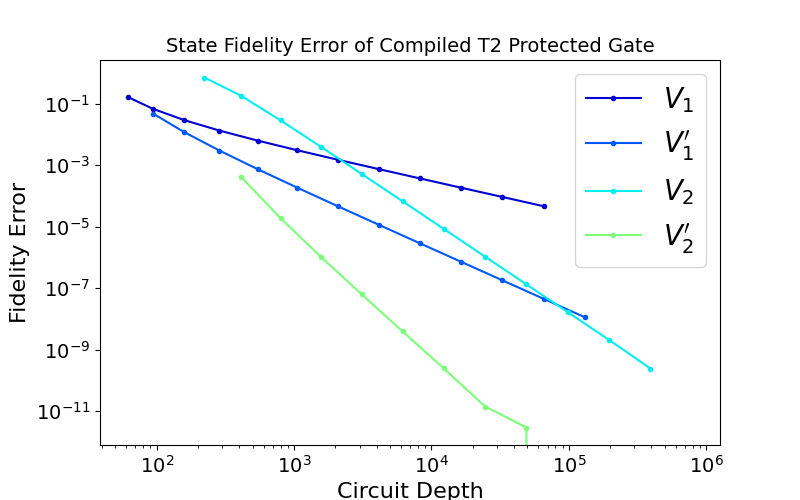}}
    \subfigure[]{\includegraphics[scale=0.42]{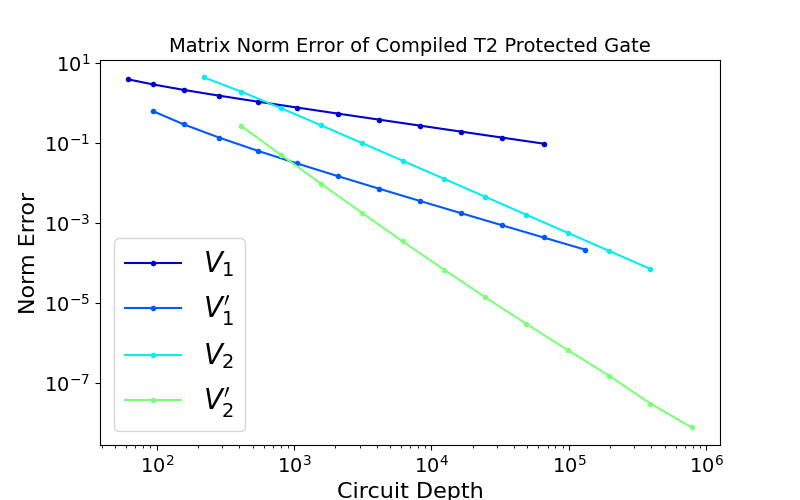}}
    }
    \caption{Error scaling of protected $T_2$ gate with respect to BCH circuit depth according to (a) the state-fidelity metric and (b) the trace norm. Prime denotes that the formula has been symmetrized. %
    }
    \label{fig:protectedt2-error}
\end{figure}

\subsection{Hong-Ou-Mandel effect/conditional (controlled-phase) beam splitter gate}\label{subsec:Conditional-(Controlled-Phase)-Beam-Splitter}
The operations we seek to realize need not act on a single \qumode; in fact, our techniques are extensible to hybrid setups with multiple \qumodes~or qubits. Consider the conditional (controlled-phase) beam splitter
\begin{align}
U_{\text{beam split}} & =e^{-i\lambda^{2}\left(\hat{a}_{1}^{\dagger}\hat{a}_{2}+\hat{a}_{1}\hat{a}_{2}^{\dagger}\right)\sigma^{z}}.
\end{align}
This gate naturally pertains
to certain lattice gauge theories \cite{C2QA_LGT} and gives rise to exponential SWAP
(eSWAP) \cite{gao2019entanglement} and controlled-SWAP (cSWAP)
gates for state purification and SWAP tests, when paired with an ordinary (uncontrolled)
beam splitter \cite{pietikainen2022controlled}. 
The argument in phase-space representation 
is 
\begin{align}
-i\lambda^{2}\left(\hat{a}_{1}^{\dagger}\hat{a}_{2}+\hat{a}_{1}\hat{a}_{2}^{\dagger}\right)\sigma^{z} & =-2i\lambda^{2}\left(\hat{x}_{1}\hat{x}_{2}+\hat{p}_{1}\hat{p}_{2}\right)\sigma^{z},
\end{align}
such that the gate is decomposed in terms of a Trotter expansion
as the product of two exponential
terms $\exp\left(\left[A_{1},B_{1}\right]\lambda^{2}\right)=\exp\left(-2i\lambda^{2}\hat{x}_{1}\hat{x}_{2}\sigma^{z}\right)$
and $\exp\left(\left[A_{2},B_{2}\right]\lambda^{2}\right)=\exp\left(-2i\lambda^{2}\hat{x}_{1}\hat{x}_{2}\sigma^{z}\right)$.
According to the Pauli commutation relation, 
the first commutator is
\begin{align}
\left[A_{1},B_{1}\right] & =-2i\hat{x}_{1}\hat{x}_{2}\sigma^{z}\\
 & =-2i\hat{x}_{1}\hat{x}_{2}\left(-\frac{i}{2}\left[\sigma^{x},\sigma^{y}\right]\right)\\
 & =\left[i\hat{x}_{1}\sigma^{x},i\hat{x}_{2}\sigma^{y}\right],
\end{align}
and the second is
\begin{align}
\left[A_{2},B_{2}\right] & =\left[i\hat{p}_{1}\sigma^{x},i\hat{p}_{2}\sigma^{y}\right],
\end{align}
with the following error scaling:

\begin{restatable*}{theorem}{resultbeamsplitter}
Assume we may implement $\text{e}^{ i t \hat{x}_m \sigma^j}, \text{e}^{ i t \hat{p}_m \sigma^j}$ for $m \in \{ 1, 2 \}$; i.e., we may implement the qubit-conditional position shifts and momentum boosts on either \qumode~without error. Then, we may approximate $\mathcal{B}_{\hat{x}_1 \hat{x}_2 + \hat{p}_1 \hat{p}_2}$ with arbitrary error scaling $p$ as
\begin{align*}
    \norm{\exp it \widetilde{\mathcal{B}}_{\hat{x}_1 \hat{x}_2 + \hat{p}_1 \hat{p}_2} - \exp it \mathcal{B}_{\hat{x}_1 \hat{x}_2 + \hat{p}_1 \hat{p}_2} } \in \mathcal{O}((Ct)^{p + \frac{1}{2}}),
\end{align*}
where $C = \max( \norm{ \hat{x}_1 \hat{x}_2 + \hat{p}_1 \hat{p}_2}, \norm{\hat{x}_1}^2, \norm{\hat{x}_2}^2, \norm{\hat{p}_1}^2, \norm{\hat{p}_1}^2)$ and using no more than $4 \cdot 5^{ \frac{p}{2} - \frac{1}{4} }$ exponentials.
\end{restatable*}

The two exponential terms are decomposed via the BCH formula 
for a lower-bound gate depth of eight. Results are shown in~\cref{fig:ConditionalBeamSplitter} for
$15$ states per cavity with a shared qubit over a final time of
$\pi/2$ with $200$ equal time
steps, where the system is initially in the first excited state of
each cavity and the ground state of the shared qubit $\left|11g\right>$.
As expected for the conditional beam splitter, the gate exhibits the
Hong-Ou-Mandel effect, in which the occupation of cavity 1 oscillates
between the first excited mode and a superposition of the ground and
the second excited states of the cavity. The BCH-synthesized results
closely agree with that of the original gate, with no visible leakage
beyond the physical states (the lowest three states of the cavity)
into the working space under the time duration studied. As for the
conditional rotation gate, the relative error of the BCH-synthesized
gate computed for a single time step of length $t$ was found to scale according to a power law with the time step,
in accordance with the analytic result for Trotterization and BCH
decomposition.

\begin{figure}[!ht]
\begin{centering}
\includegraphics[width=0.5\columnwidth]{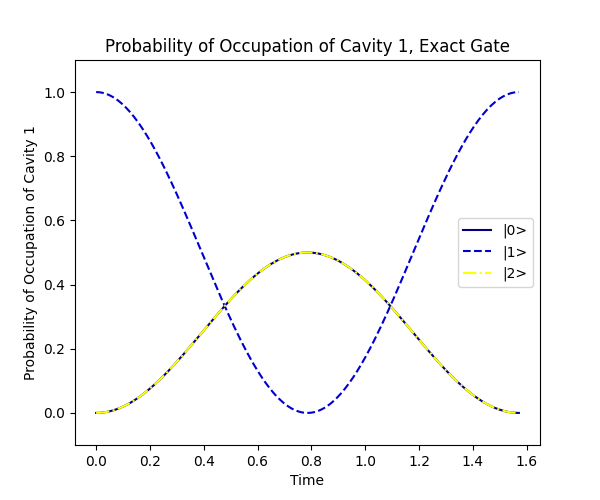}\includegraphics[width=0.5\columnwidth]{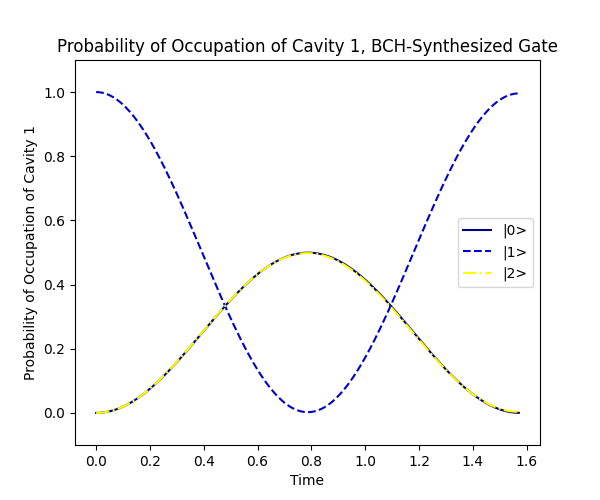}
\par\end{centering}
\begin{centering}
\includegraphics[width=0.5\columnwidth]{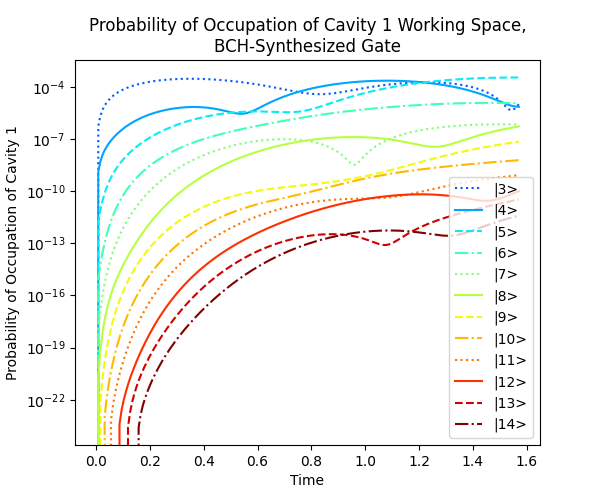}\includegraphics[width=0.5\columnwidth]{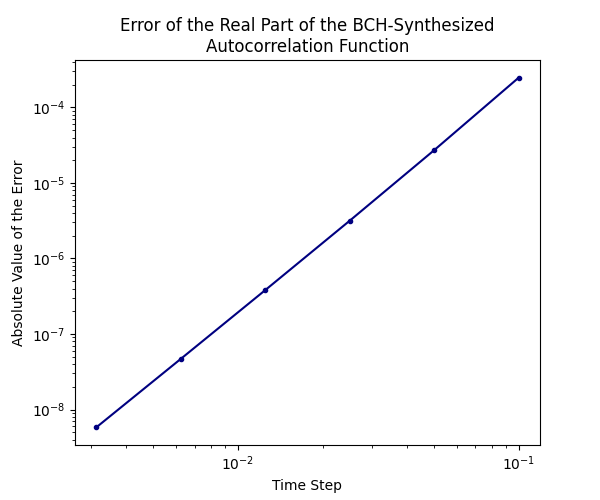}
\par\end{centering}
\caption{Hong-Ou-Mandel effect simulated with (a) exact and (b) BCH-synthesized
conditional beam splitters, illustrated as probability cavity 1 is found in states $\left|0\right>$, $\left|1\right>$, or $\left|2\right>$. \change{Observe that the probabilities of $\ket{0}, \ket{2}$ coincide}; (c) probability of leakage into higher cavity modes; and (d) error of the real part of the autocorrelation
after a single application of a BCH-synthesized gate for time step $t$ relative to the application of the exact gate for the same time step.\label{fig:ConditionalBeamSplitter}}%
\end{figure}

{

}

\section{Discussion}

Our main contribution in this paper is a systematic approach to synthesizing unitary dynamics on a hybrid quantum computer that has access to both qubit and \qumode~operations.  Such gate sets naturally model systems such as cavity quantum electrodynamics systems and ion trap-based quantum computers.  Our main innovation here is the development of high-order analytic formulas that can be used to place bounds on the complexity of implementing arbitrary unitary operations on such a hybrid device.  Specifically, we see that these methods are capable of achieving subpolynomial scaling with the inverse error tolerance ($1/\epsilon$) and allow us to implement arbitrary nonlinearities in the field operators in the generator of the unitary that we wish to implement at low cost asymptotically.  In particular, we focus on using a construct known as a block-encoded creation operator as our fundamental construct and show numerically highly accurate approximations to the exponential of a block encoding of the square of the creation operator.  Further, we study the Hong-Ou-Mandel effect and observe that the synthesized operations used in our construction can have negligible error with respect to our target precision.
Beyond compilation, our work thereby provides intuition when compiling a variety of gates. For example, our methods give mathematical understanding for the recent successes of echoed conditional  displacement (ECD) gate sets for quantum computing \cite{eickbusch2021fast}.

While this work enables better analytic control of qubit-qumode systems, there remain many open questions.  
In particular, lowering the resource cost of these compilations, especially in noisy environments, is key to practical implementations.
We note that for both the Hong-Ou-Mandel effect and the block encoding of $(a^\dagger)^2$ that thousands of gate operations are needed to achieve infidelities of $10^{-3}$ or smaller.  This makes such sequences impractical for near-term applications where the gate infidelities are on the order of $1\%$. Ideally, these analytic sequences can be used as an initial sequence which can be further optimized via existing optimal control techniques. This would reduce warmstart challenges and potentially improve convergence.

Another concern arises from the nature of the qubit system, namely that creating qubit-qumode interactions requires truncating the qubit system to a two-level system. For example, transmon qubits are typically truncated to two levels~\cite{wendin2017quantum,kjaergaard_superconducting_2020,blais2021circuit};  in reality, these transmons can experience leakage to higher states, meaning that the qubit may itself exhibit qumode-like properties when poorly calibrated.
While leakage is a practical concern, it can be somewhat avoided via well-defined pulse shaping methods (e.g., DRAG~\cite{theis_counteracting_nodate} pulses). These methods allow one to approach the control speed limit set by the anharmonicity of the level structure.

Several avenues of approach exist that could be used to improve upon these results: first, note that our technique is connected to ideas from quantum signal processing (QSP) \cite{martyn_grand_2021, gilyen_quantum_2019}. Thus, QSP could be applied to approximate polynomials of qumode operators; this could potentially improve the scaling with respect to the error tolerance from the subpolynomial scaling currently demonstrated to polylogarithmic scaling. Alternatively, our approach could seed gradient-descent optimization procedures for control such as GRAPE~\cite{khaneja2005optimal} at the pulse level or numerical optimization of parameterized gates at the SNAP~\cite{fosel2020efficient} or ECD \cite{eickbusch2021fast} instruction level.  These locally optimized sequences may then prove to be either better, or more understandable, than existing gradient-optimized pulse sequences for control of such systems. %
Designing pulses is, in practice, sensitive to device properties and noise environments,
and presents a rich area for further research.

\section*{Author Contributions}
\change{SMG and NW conceived of the instruction set architecture and applications. CK derived and obtained analytic error bounds, developed the state preparation subroutine, and mathematically formalized the scheme. MBS developed compilation strategies for practical implementation, investigated the efficacy of the approach for physical applications, and performed numerical simulations of gate compilation and quantum dynamics. EC advocated for clarified exposition across all applications. All authors contributed to the manuscript, with EC, CK, and MBS producing the initial and revised manuscripts.}

\section*{Acknowledgements}
This project was primarily supported by the U.S. Department of
Energy, Office of Science, National Quantum Information
Science Research Centers, Co-design Center for Quantum
Advantage under contract number DE-SC0012704, (Basic Energy Sciences, PNNL FWP 76274). C\textsuperscript{2}QA led this research. Christopher Kang was also funded in part by the STAQ project under award NSF Phy-232580; in part by the US Department of Energy Office of Advanced Scientific Computing Research, Accelerated 
Research for Quantum Computing Program; and in part by the NSF Quantum Leap Challenge Institute for Hybrid Quantum Architectures and Networks (NSF Award 2016136), in part based upon work supported by the U.S. Department of Energy, Office of Science, National Quantum 
Information Science Research Centers, and in part by the Army Research Office under Grant Number W911NF-23-1-0077. Micheline B.~Soley was supported by the Yale Quantum Institute Postdoctoral Fellowship and the Office of the Vice Chancellor for Research and Graduate Education at the University of Wisconsin-Madison with funding from the Wisconsin Alumni Research Foundation. Eleanor Crane was supported by UCL Faculty of Engineering Sciences and the Yale-UCL exchange scholarship from
RIGE (Research, Innovation and Global Engagement),
the Princeton MURI award SUB0000082, the DoE QSA, NSF
QLCI (award No.OMA-2120757), DoE ASCR Accelerated
Research in Quantum Computing program (award No.DESC0020312), and NSF PFCQC program. 

The views and conclusions contained in this document are those of the authors and should not be interpreted as representing the official policies, either expressed or implied, of the U.S. Government. The U.S. Government is authorized to reproduce and distribute reprints for Government purposes notwithstanding any copyright notation herein.

External interest disclosure: SMG is a consultant for and equity holder in Quantum Circuits, Inc.

\appendix

\newpage
\appendix
\section*{APPENDICES}
\section{Obtaining \texorpdfstring{$S_1$}{}}\label{obtaining_s1}

We demonstrate how to obtain the $\mathcal{S}_1$ operator with the native gates present in the dispersive coupling regime of the Jaynes-Cummings model. Astute readers may recognize the $\mathcal{S}_1$ operator as being a Jaynes-Cummings Hamiltonian when rewritten in the form
\begin{align}
    \ket{0} \bra{1} \otimes a^\dagger + \ket{1} \bra{0} \otimes a.
\end{align}
It is thus unsurprising that we may implement these types of operations on a circuit QED device that exhibits a similar native Hamiltonian.

The conditional  displacement operator with magnitude $\alpha$ is written:
\begin{equation}
    U_{d}(\alpha) = e^{i \sigma^z (\alpha a^\dagger + \alpha^* a) }.
\end{equation}
For a Fock state $\ket{n}$, $a^\dagger a \ket{n} = n \ket{n}$; therefore $e^{i a^\dagger a \theta} a^\dagger e^{-i a^\dagger a \theta}=e^{i\theta}a^\dagger$. Taking $\alpha=\alpha^*$, we have
\begin{equation}
    e^{i(\pi/2) a^\dagger a} \cdot e^{i\sigma^y(\alpha(a^\dagger + a))} \cdot e^{-i(\pi/2) a^\dagger a} = \exp\left(\alpha \begin{bmatrix}0 & (ia^\dagger - ia)\\ (-ia^\dagger + ia) &0 \end{bmatrix} \right).
\end{equation}
This operation can be built using single-qubit operations on a controlled displacement gate with additional phase delays on the oscillator.  Both are linear optical operations or single-qubit operations, which we expect to be inexpensive in our computational model.

Next, note that 
\begin{equation}
    e^{i(\alpha(a^\dagger + a))\otimes \sigma^x} = \exp\left(\alpha\begin{bmatrix} 0 & i(a^\dagger + a) \\ i(a^\dagger +a) &0 \end{bmatrix} \right).
\end{equation}
Thus to $O(\alpha^2)$ the block-encoded creation operation can be constructed using single-qubit, controlled-displacement, and linear-optical operations through

\begin{equation}
    \mathcal{S}_1 \approx e^{i(\pi/2) a^\dagger a}e^{i(\alpha(a^\dagger + a))\otimes \sigma^y} e^{-i(\pi/2) a^\dagger a}e^{i(\alpha(a^\dagger + a))\otimes \sigma^x}.
\end{equation}
\cref{fig:obtaining-s1} describes the circuit.

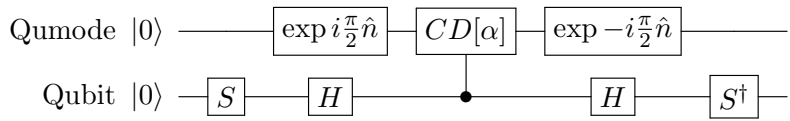
\begin{figure}[h]
    \centering
    \[
    \Qcircuit @C=1em @R=.7em {
    \lstick{\textrm{Qumode } \ket{0}} & \qw & \gate{\exp i \frac{\pi}{2} \hat{n}} & \gate{CD[\alpha]} & \gate{\exp - i \frac{\pi}{2} \hat{n}} & \qw & \qw\\
    \lstick{\textrm{Qubit } \ket{0}} & \gate{S} & \gate{H} & \ctrl{-1} & \gate{H} & \gate{S^\dagger} & \qw
    }
    \]
    \caption{Circuit to approximate $\mathcal{S}_1$ with conditional displacement gates on the qumode-qubit device.}
    \label{fig:obtaining-s1}
\end{figure}

\section{Error analysis}\label{apndx:error-analysis}
To assess the error scaling of our approach, %
we must consider three sources of error: the underlying implementation error from using approximations of $\mathcal{B}_A(t)$, the error from BCH, and the error from Trotter. In \cref{apndx:implementation-error-bch-trotter}, we show that the BCH and Trotter formulas can still be applied on exponentials that have error. Then, we use these formulas in \cref{apndx:general-addition-error} to produce the error bounds for addition. Finally, in \cref{apndx:multiplication}, we produce the error bounds for multiplication.

\subsection{Product formulas with implementation error}\label{apndx:implementation-error-bch-trotter}
We begin by formally stating the Trotter and BCH formulas when there is no implementation error:
\begin{theorem}[BCH Product Formula (Theorem 2 from \cite{Childs_2013})]\label{fact:bch}
Let $A$ and $B$ be bounded complex-valued matrices and assume without loss of generality $t \in \mathbb{R}^+$ is assumed for the purposes of asymptotic analysis to be in $o(1)$.  We then define
\begin{align}
    \bch_{1, k} (At, Bt^k) \coloneqq e^{At} e^{Bt^k} e^{-At} e^{-Bt^k}.
\end{align}
We further define $\bch_{p,k}$ recursively for $p \geq 2$ and odd $ k \geq 1$:
\begin{align}
    \bch_{p + 1, k}(At, Bt^k) &\coloneqq \bch_{p, k} (A \gamma_p t, B (\gamma_p t)^k ) \bch_{p, k} (-A\gamma_p t, - B (\gamma_p t)^k) \\
    &\times \bch_{p, k} (A \beta_p t, B (\beta_p t)^k )^{-1} \bch_{p, k} (-A \beta_p t, - B (\beta_p t)^k )^{-1} \\
    &\times \bch_{p, k} (A \gamma_p t, B (\gamma_p t)^k ) \bch_{p, k} (-A \gamma_p t, - B (\gamma_p t)^k ),
\end{align}
with the constants
\begin{align}
    \beta_p \coloneqq (2r_p )^{1/(k + 1)}, \gamma_p \coloneqq (1/4 + r_p )^{1 / (k + 1)}, r_p \coloneqq \frac{2^{\frac{(k + 1)}{2p + k + 1}}}{4 \left(2 - 2^{\frac{k + 1}{2p + k + 1}} \right)}.
\end{align}
This recursive formula has the error scaling, where $\gamma = \max(\norm{A}, \norm{B}^{1/k})$~\cite{Childs_2013}
\begin{align}
    \bch_{p, k}(At, Bt^k) = e^{[A, B]t^{k + 1} + \mathcal{O}((\gamma t)^{2p + k})}
\end{align}
and uses $8 \cdot 6^{p - 1}$ exponentials when $k = 1$ and $4 \cdot 6^{p-1}$ exponentials otherwise. 
\end{theorem}

\begin{theorem}[Trotter Formula (Lemma 1 from \cite{berry2007efficient})]\label{fact:trotter}
Let $\{H_j: j = 1\ldots m \}$ be a set of $M$ bounded Hermitian operators acting on a Hilbert space of dimension $2^n$ and assume without loss of generality that $t\ge 0$. For $H= \sum H_j$, the error in the Trotter-Suzuki formulas of order $k$ and timestep $r$ obeys the error bound
\begin{align}
    \norm{\exp\left( -it \sum_{j = 1}^m H_j \right) - \trotter_{ 2k}(\{ H_j \}, t / r)^r  } \leq 5 (2 \times 5^{k - 1} m \tau)^{2k + 1} / r^{2k},
\end{align}
where $\tau = \norm{H} t$ and
\begin{align}\label{req:trotter-constraints}
    4 m 5^{k - 1} \tau / r &\leq 1 , \\
    (16/3) ( 2 \times 5^{k - 1} m \tau)^{2k + 1} / r^{2k } &\leq 1 ,
\end{align}
using no more than $2 m 5^{k - 1} r$ exponentials. We define the $k^\text{th}$ order Trotter formula as
\begin{align*}
    \trotter_{2k} (\lambda) := [\trotter_{2k - 2} (p_k \lambda)]^2 \trotter_{2k - 2} ((1 - 4p_k)\lambda) [\trotter_{2k - 2} (p_k \lambda)]^2,
\end{align*}
where $p_k = (4 - 4^{1/(2k - 1)})^{-1}$ and $k > 1$. The relation has the base case
\begin{align*}
    \trotter_2(\lambda) = \prod_{j = 1}^m e^{H_j \lambda / 2} \prod_{j' = m}^1 e^{H_{j'} \lambda / 2}.
\end{align*}
\end{theorem}
This implies the following corollary: 
\begin{corr}\label{corr:trotter-r}
If $r = 1$, i.e.~there is no time stepping, the Trotter formula exhibits the error scaling
\begin{align}
    \norm{\exp \left(-it \sum_{j = 1}^m H_j \right) - \trotter_{2k}(\{H_j\}, t)} \in \mathcal{O}((\norm{H} t)^{2k + 1}),
\end{align}
using no more than $2 m 5^{k - 1}$ exponentials.
\end{corr}

\newcommand{\uone}{U_1(t)}
\newcommand{\utwo}{U_2(t)}
\newcommand{\uonetilde}{\widetilde{U}_{1, p_1}(t)}
\newcommand{\utwotilde}{\widetilde{U}_{2, p_2}(t)}
\newcommand{\aonetilde}{\Tilde{A}_{1, p_1}(t)}
\newcommand{\atwotilde}{\Tilde{A}_{2, p_2}(t)}

Both of these formulas, however, assume that our implementation of exponentials occurs without error. However, if we would like to apply our technique recursively, our matrix product formulas must account for implementation error, i.e.~be able to use primitives that themselves may have error. Thus, we restate both Trotter and BCH when the operators have asymptotic error:
{
\begin{lemma}[BCH under implementation error]\label{lem:bch-with-approx}
Suppose there are some ideal operators $\uone, \utwo$ that are exponentials of some anti-Hermitian matrix, i.e.: 
\begin{align}
    \uone &= \exp t A_1, \\
    \utwo &= \exp t A_2.
\end{align}
We seek to build $\exp t^2 [A_1, A_2]$, the exponential of the commutator of the matrices. Also suppose that may approximate $\uone, \utwo$ with $\uonetilde, \utwotilde$ with the error scaling
\begin{align}
    \norm{\uonetilde - \uone} &\in \mathcal{O}((c t)^{p_1}), \\
    \norm{\utwotilde - \utwo} &\in \mathcal{O}((c t)^{p_2}),
\end{align}
for some $p_1, p_2 \geq 1$. Call the logarithms of $\Tilde{U}_{1, p_1}, \Tilde{U}_{2, p_2}$ to be $\aonetilde, \atwotilde$. Then, applying a $q^\text{th}$-ordered BCH formula, where $q = \max(\ceil{\frac{\min(p_l, p_2) - 1}{2}}, 1 )$ on the implementable $\uonetilde, \utwotilde$ can still approximate the commutator  exponential by applying \cref{fact:bch}:
\begin{align}
    \norm{\exp t^2 [A_1, A_2] - \bch_{q, r}(\aonetilde, \atwotilde)} \in \mathcal{O}((C t)^{\min (p_1, p_2)}),
\end{align}
where $C = \max(\norm{A_1}, \norm{A_2}, c)$. This procedure uses $8 \cdot 6^{q - 1}$ total exponentials.

\end{lemma}
\begin{proof}
Recognize that we may decompose the error involved in implementing the commutator exponential into two sources: the error incurred from the BCH formula intrinsically and the implementation error from the realizable terms. Thus, by the triangle inequality
\begin{align}
    &\norm{\exp t^2 [A_1, A_2] - \bch_q(\aonetilde, \atwotilde)} \leq \nonumber\\
    &\; \norm{\exp t^2 [A_1, A_2] - \bch_q(A_1, A_2)} + \norm{\bch_q(A_1, A_2) - \bch_q(\aonetilde, \atwotilde)} .
\end{align}
We begin with the LHS term. By \cref{fact:bch}, 
\begin{align}
    \norm{\exp t^2 [A_1, A_2] - \bch_q(A_1, A_2)} \in \mathcal{O}((C_{\textrm{BCH}} t)^{2q + 1}),
\end{align}
where $C_{\textrm{BCH}} = \max(\norm{A_1}, \norm{A_2})$. 
For the RHS, recall by Box 4.1 of Nielsen and Chuang \cite{nielsen_chuang_2010} that implementation errors accumulate at most linearly; thus, we can sum over the $8 \cdot 6^{q- 1}$ operations used by BCH. By symmetry of the BCH formula, we apply the $\uonetilde, \utwotilde$ exponentials precisely $4 \cdot 6^{q- 1}$ times. Thus,
\begin{align}
    \norm{\bch_q(A_1, A_2) - \bch_q(\aonetilde, \atwotilde)} &\leq \sum_{j = 1}^{4 \cdot 6^{q - 1}} \mathcal{O}((c t)^{p_1}) + \mathcal{O}((c t)^{p_2}) \\ 
    &\in \mathcal{O}((c t)^{\min(p_l, p_r)}).
\end{align}
Setting $q = \max\{\ceil{\frac{\min(p_l, p_r) - 1}{2}}, 1 \} $ and recalling $\gamma \geq 1$, we observe
\begin{align}
    \norm{\exp t^2 [A_1, A_2] - \bch_q(\aonetilde, \atwotilde)} \in \mathcal{O}((C t)^{\min(p_l, p_r)}),
\end{align}
as desired.
\end{proof}
\begin{lemma}[Trotter under implementation error]\label{lem:trotter-with-approx}
Given~\cref{lem:bch-with-approx}'s assumptions, \cref{fact:trotter}'s assumptions, and assuming $\norm{A_1 + A_2} \geq 1$, an operator can be constructed by applying a $q^\text{th}$-ordered Trotter formula $\trotter_{2q}$ by setting $q = \max(\ceil{\frac{\min(p_l, p_2) - 1}{2}}, 1 )$ so that
\begin{align}
    \norm{\exp t (A_1 + A_2) - \trotter_q (\aonetilde, \atwotilde)} \in \mathcal{O}((C t)^{\min(p_1, p_2)}),
\end{align}
where $C = \max(\norm{A_1 + A_2}, c)$ and using no more than $4 \cdot 5^{q -1} $ operator exponentials. 
\end{lemma}
\begin{proof}
Similarly, we may apply the triangle inequality in order to determine a bound by separating the error accrued into the intrinsic Trotter error and the implementation error
\begin{align}
    &\norm{\exp t (A_1 + A_2) - \trotter_{2q} (\aonetilde, \atwotilde) } \nonumber\\
    &\leq \norm{\exp t(A_1 + A_2) - \trotter_{2q} (A_1, A_2)}\nonumber \\
    &\qquad+ \norm{\trotter_{2q} (A_1, A_2) - \trotter_{2q} (\aonetilde, \atwotilde)}.
\end{align}
To analyze the LHS, which represents the Trotter error, \cref{corr:trotter-r} provides a bound
\begin{align}
    \norm{\exp t(A_1 + A_2) - \trotter_{2q} (A_1, A_2)} \in \mathcal{O}((\norm{A_1 + A_2}t)^{2q + 1}).
\end{align}
To analyze the RHS, which represents the implementation error, recall by box 4.1 of Nielsen and Chuang \cite{nielsen_chuang_2010} that the error accrues linearly. Furthermore, the number of operations in Trotter is no more than $4 \cdot 5^{q - 1} $ operations in total. Therefore, we only apply each constituent operation at most $2 \cdot 5^{q - 1}$ times. Thus, a loose upper bound can be written as
\begin{align}
    \norm{\trotter_{2q} (\uone, \utwo) - \trotter_{2q} (\aonetilde, \atwotilde)} &\leq \sum_{j = 1}^{4 \cdot 5^{q -1 }} \mathcal{O}((c t)^{p_1}) + \mathcal{O}((c t)^{p_2}) \\
    &\in \mathcal{O}((c t)^{\min(p_1, p_2)}).
\end{align}
It is sufficient to set $q = \max(\ceil{\frac{\min(p_1, p_2) - 1}{2}}, 1)$ so that
\begin{align}
    &\norm{\exp t (A_1 + A_2) - \trotter_{2q} (\aonetilde, \atwotilde) } \\
    &\qquad \in \mathcal{O}((\norm{A_1 + A_2} t)^{2q + 1}) + \mathcal{O}(( c t)^{\min(p_1, p_2)}) =  \mathcal{O}((Ct)^{\min(p_1, p_2)}),
\end{align}
as desired.
\end{proof}
}

\subsection{Scaling of the addition algorithm}\label{apndx:general-addition-error}
We apply the above results to produce the error analysis of \cref{alg:adder}:
\algproduct*
\begin{proof}
Our proof proceeds by applying the above theorems upon our operations. By setting $q = \max(\ceil{\frac{\min(p_1, p_2) - 1}{2}}, 1)$, \cref{lem:bch-with-approx} implies that:
\begin{align}
    \norm{\textrm{BCH}_q ( X \cdot  \widetilde{\mathcal{B}}_B(\tau) \cdot X, \widetilde{\mathcal{B}}_A(\tau)) -  \exp \tau^2 [A, B] \sigma^z} &\in \mathcal{O}((C_{BCH}\tau)^{p_{BCH}}), \label{eq:C}\\
    \norm{\textrm{BCH}_q (S\cdot  \widetilde{\mathcal{B}}_A(\tau) \cdot  S^\dagger,  X  \cdot  \widetilde{\mathcal{B}}_B(\tau) \cdot  X) -  \exp i \tau^2 \{A, B \} \sigma^z} &\in \mathcal{O}((C_{BCH} \tau)^{p_{BCH}}), \label{eq:AC}
\end{align}
where $C_{BCH} = \max( \norm{A}, \norm{B}, c)$ and $p_{BCH} = \min (p_A, p_B) $. We'll need to set $\tau = \sqrt{\frac{t}{2}}$ to achieve the desired time evolution. Additionally, Pauli conjugation has no impact on the error scaling. We call these formulas ``Left" and ``Right"
\begin{align}
    \textrm{LEFT} &\coloneqq SH \cdot \textrm{BCH}_q ( X\cdot  \widetilde{\mathcal{B}}_B(\tau)\cdot  X, \widetilde{\mathcal{B}}_A(\tau)) \cdot HS^\dagger, \\
    \textrm{RIGHT} &\coloneqq H \cdot  \textrm{BCH}_q (S\cdot  \widetilde{\mathcal{B}}_A(\tau)\cdot  S^\dagger, X\cdot \widetilde{\mathcal{B}}_B(\tau) \cdot X )\cdot  H,
\end{align}
where
\begin{align}
    \norm{\textrm{LEFT} - \exp \tau^2 \sigma^y (AB - (AB)^\dagger) } &\in \mathcal{O}((C^2 t)^{p_{BCH} / 2}), \\
    \norm{\textrm{RIGHT} - \exp i \tau^2 \sigma^x (AB + (AB)^\dagger)   } &\in \mathcal{O}((C^2 t)^{p_{BCH} / 2}).
\end{align}
Finally, we use \cref{lem:trotter-with-approx}, implying that a Trotter formula with order $q$ has the error scaling
\begin{align}
    \norm{\textrm{Trotter}_s (\textrm{LEFT}, \textrm{RIGHT})  - \exp it \begin{bmatrix}
        0 & (AB)^\dagger \\
        AB & 0 
    \end{bmatrix} } \in \mathcal{O}((C_{TOTAL} t)^{p_{BCH} / 2}),
\end{align}
where
\begin{align}
    C_{TOTAL} = \max \left( \norm{\begin{bmatrix}
        0 & (AB)^\dagger \\
        AB & 0
    \end{bmatrix}}, C_{BCH}^2 \right) = \max \left(\norm{AB}, \norm{(AB)^\dagger}, C_{BCH}^2 \right).
\end{align}

To bound the number of operations used, we recognize that the Trotter formula requires at most $4 \cdot 5^{q - 1}$ commutators, each of which requires $8 \cdot 6^{q - 1}$ constituent operators. Thus, the total number of operators required is bounded as
\begin{align}
    4 \cdot 5^{q - 1} \cdot 8 \cdot 6^{q - 1} \leq 1.07 \cdot 30^q ,
\end{align}
\end{proof}

\newcommand{\leftarbop}{\widetilde{\mathcal{S}}_{k_l, p_l}(t)}
\newcommand{\leftarbopt}{\widetilde{\mathcal{S}}_{k_l, p_l}(\tau)}
\newcommand{\leftarbopx}{\widetilde{\mathcal{S}}^X_{k_l, p_l}(\tau)}
\newcommand{\leftarbopy}{\widetilde{\mathcal{S}}^Y_{k_l, p_l}(\tau)}

\newcommand{\rightarbop}{\widetilde{\mathcal{S}}_{k_r, p_r}(t)}
\newcommand{\rightarbopt}{\widetilde{\mathcal{S}}_{k_r, p_r}(\tau)}
\newcommand{\rightarbopx}{\widetilde{\mathcal{S}}^X_{k_r, p_r}(\tau)}
\newcommand{\rightarbopxr}{\widetilde{\mathcal{S}}^X_{k_r, p_r}(\tau / r_{\textrm{BCH})}}
\newcommand{\rightarbopy}{\widetilde{\mathcal{S}}^Y_{k_r, p_r}(\tau)}
\newcommand{\rightarbopyr}{\widetilde{\mathcal{S}}^Y_{k_r, p_r}(\tau / r_{\textrm{BCH}})}
This implies the following corollary for the annihilation/creation operators:
\begin{corr}[\cref{alg:adder} applied to polynomials of annihilation/creation operators]\label{lem:adder}
Assume we can implement the following $k_l, k_r^\text{th}$ order approximations of $S$ with error scaling $p_l, p_r$:
\begin{align}
    \norm{\leftarbop - \mathcal{S}_{k_l}(t)} &\in \mathcal{O}((c t)^{p_l}), \\
    \norm{\rightarbop - \mathcal{S}_{k_r}(t)} &\in \mathcal{O}((c t)^{p_r}),
\end{align}
with $c \geq \Lambda^{\max(k_l, k_r) / 2}$. Then, we can implement higher-order operators with comparable $t$ scaling
\begin{align}
    \norm{
    \widetilde{\mathcal{S}}_{k_l + k_r, \min(p_l, p_r)}(t)  - 
    \exp \left(i t \begin{bmatrix}
    0 & ( a^\dagger )^{k_l + k_r} \\
    a^{k_l + k_r} & 0
    \end{bmatrix}
    \right)
    } \in \mathcal{O}((c^2 t)^{\min(p_l, p_r) / 2}),
\end{align}
using no more than $1.07 \cdot 30^q$ $\mathcal{S}_{k_l},  \mathcal{S}_{k_r}$ operators.
\end{corr}
\begin{proof}
To synthesize the block encoding of higher-order annihilation/creation operators, we directly apply 
 \cref{thm:general-adder-error}. We quantify the error by bounding the block-encoding norm. Note that:
    \begin{fact}\label{fact:norm-blockencodedxt}
A $k^\text{th}$ order block-encoded operator has a bounded norm
\begin{align}
    \norm{\begin{bmatrix}
        0 & ( a^\dagger )^k \\
        ( a )^k & 0
    \end{bmatrix}}  \leq \Lambda^{k / 2}.
\end{align}
\end{fact}
Thus, the constant $C_\textrm{BCH} $ is bounded, as $C_\textrm{BCH} \leq \max(\Lambda^{k_l / 2}, \Lambda^{k_r / 2}, c) \leq \max( \Lambda^{\max(k_l, k_r) / 2}, c) = c$ by hypothesis. Next, observe that $\norm{AB}, \norm{(AB)^\dagger} \leq \Lambda^{(k_l + k_r) / 2} $. Thus, $C_{TOTAL} \leq \max( \Lambda^{(k_l + k_r) / 2}, c^2) = c^2$. Therefore, the final error scaling will be upper bounded by $\mathcal{O}((c^2 t )^{\min(p_l, p_r) / 2})$.
\end{proof}

\subsection{Scaling of the multiplication algorithm}\label{apndx:multiplication}
\cref{alg:mult}'s error scaling follows directly from the BCH formula:
\algmult*
\begin{proof}
    We can directly apply \cref{lem:bch-with-approx} using the $\widetilde{\mathcal{B}}_A, \widetilde{\mathcal{B}}_B$ operators. When applied, we find that
    \begin{align}
        \norm{ \textrm{BCH}_q( S \cdot i \tau \widetilde{\mathcal{B}}_A(\tau) \cdot S^\dagger, X \cdot i \tau \widetilde{\mathcal{B}}_B \cdot X) - \exp 2 i \tau^2 \begin{bmatrix}
            AB & 0 \\
            0 & - BA - (BA)^\dagger
        \end{bmatrix}} \in \mathcal{O}((C \tau)^{\min(p_A, p_B)}),
    \end{align}
    where $C = \max (\norm{A}, \norm{B}, c)$ and $q = \max ( \ceil{\frac{\min(p_A, p_B) - 1}{2}}, 1) $. Then, by taking $\tau = \sqrt{\frac{t}{2}}$, we yield
    \begin{align}
        \norm{\textrm{BCH}_q(S\cdot  \widetilde{\mathcal{B}}_B(\tau)\cdot  S^\dagger, \widetilde{\mathcal{B}}_A(\tau) ) - \exp i t \begin{bmatrix}
            AB & 0 \\
            0 & -BA
        \end{bmatrix}} \in \mathcal{O}((C \tau)^{\min(p_A, p_B)}) = \mathcal{O}((C^2 t)^{\min(p_A, p_B) / 2}).
    \end{align}
    By counting the number of exponentials in the result via \cref{lem:bch-with-approx}, we finally find that the number of exponentials needed is at most $8 \cdot 6^{q - 1}$. 
\end{proof}

\section{Phase-Space Applications}
In the following section, we derive error bounds for the two phase-space applications described: the conditional rotation gate and the controlled-phase beam splitter. Note that the cutoff approach we employ to bound $\norm{a}, \norm{a^\dagger}$ is more complex for position and momentum operators.

To obtain error bounds, we leave all expressions in terms of $\norm{\hat{x}}, \norm{\hat{p}}$. We leave a more concrete bound that can be found by applying a cutoff upon both $\hat{x}, \hat{p}$ simultaneously to future work. 

\subsection{Conditional rotation gate}\label{subsec:cond-rot}

\resultmeasurement
\begin{proof}
    Begin by directly applying \cref{fact:bch} to identify the error scaling. This implies that
    \begin{align}
        \norm{\textrm{BCH}_p(i \tau \hat{x} \sigma^i, i \tau \hat{x} \sigma^j) - \exp \tau^2 \hat{x}^2 [\sigma^i, \sigma^j] } &\in \mathcal{O}((\norm{\hat{x}} \tau)^{2p + 1}), \\
        \norm{\textrm{BCH}_p(i \tau \hat{p} \sigma^i, i \tau \hat{p} \sigma^j) - \exp \tau^2 \hat{p}^2 [\sigma^i, \sigma^j] } &\in \mathcal{O}((\norm{\hat{p}} \tau)^{2p + 1}).
    \end{align}
    Without loss of generality, select $\sigma^i = \sigma^y$ and $\sigma^j = \sigma^z$ so that $[\sigma^i, \sigma^j] = 2i \sigma^x$. Then, by selecting $\tau = \sqrt{\frac{t}{2}}$, the BCH formula is an approximation of  $\exp it \mathcal{B}_{\hat{x}^2} $. Thus,
    \begin{align}
        \norm{\exp it \widetilde{ \mathcal{B}}_{\hat{x}^2} - \exp i t \hat{x}^2 \sigma^x} &\in \mathcal{O}((\norm{\hat{x}}^2 t)^{p + \frac{1}{2}}).
    \end{align}
    Similarly, for $\hat{p}^2$,
    \begin{align}
        \norm{\exp it \widetilde{ \mathcal{B}}_{\hat{p}^2} - \exp i t \hat{p}^2 \sigma^x} &\in \mathcal{O}((\norm{\hat{p}}^2 t)^{p + \frac{1}{2}}).
    \end{align}
    We apply \cref{lem:trotter-with-approx} and observe
    \begin{align}
        \norm{\textrm{Trotter}_q( \widetilde{it \mathcal{B}}_{\hat{x}^2}, it \widetilde{\mathcal{B}}_{\hat{p}^2}) - \exp i t (\hat{x}^2 + \hat{p}^2)\sigma^x } \in \mathcal{O}( (C t)^{p + \frac{1}{2}}),
    \end{align}
    where $C = \max( \norm{ \hat{x}^2 + \hat{p}^2 }, \norm{\hat{x}}^2, \norm{\hat{p}}^2)$ and $q = \ceil{\frac{p}{2} - \frac{1}{4}}$. This requires no more than $4 \cdot 5^{q - 1}$ operator exponentials, thus implying
    \begin{align}
        4 \cdot 5^{q - 1} \leq 4 \cdot 5^{ \frac{p}{2} - \frac{1}{4} }.
    \end{align}
\end{proof}

\subsection{Controlled-phase beam splitter gate}\label{subsec:controlled-phase}
\resultbeamsplitter
\begin{proof}
We may take a similar approach as above. Applying \cref{fact:bch} yields
\begin{align}
    \norm{\textrm{BCH}_p( i \tau \hat{x}_1 \sigma^i, i \tau \hat{x}_2 \sigma^j) - \exp \tau^2 \hat{x}_1 \hat{x}_2 [\sigma^i, \sigma^j] } &\in \mathcal{O}((\norm{\hat{x}} \tau)^{2p + 1}), \\
    \norm{\textrm{BCH}_p(i \tau \hat{p}_1 \sigma^i, i \tau \hat{p}_2 \sigma^j) - \exp \tau^2 \hat{p}_1 \hat{p}_2 [\sigma^i, \sigma^j]  } &\in \mathcal{O}((\norm{\hat{p}} \tau)^{2p + 1}),
\end{align}
where we set $\norm{\hat{x}} = \max( \norm{\hat{x}_1}, \norm{\hat{x}_2})$ and $\norm{\hat{p}} = \max(\norm{\hat{p}_1}, \norm{\hat{p}_2})$. We again take $\tau = \sqrt{\frac{t}{2}}$ so that
\begin{align}
    \norm{\exp it \widetilde{ \mathcal{B}}_{\hat{x}_1 \hat{x}_2} - \exp i t \hat{x}_1 \hat{x}_2 \sigma^x} &\in \mathcal{O}((\norm{\hat{x}}^2 t)^{p + \frac{1}{2}}), \\
    \norm{\exp it \widetilde{ \mathcal{B}}_{\hat{p}_1 \hat{p}_2} - \exp i t \hat{p}_1 \hat{p}_2 \sigma^x} &\in \mathcal{O}((\norm{\hat{p}}^2 t)^{p + \frac{1}{2}}) .
\end{align}
Applying \cref{lem:trotter-with-approx} gives
\begin{align}
    \norm{\textrm{Trotter}_q( \widetilde{\mathcal{B}}_{\hat{x}_1 \hat{x}_2}, \widetilde{\mathcal{B}}_{\hat{p}_1 \hat{p}_2}) - \exp i t (\hat{x}_1 \hat{x}_2 + \hat{p}_1 \hat{p}_2)\sigma^x } \in \mathcal{O}( (C t)^{p + \frac{1}{2}}),
\end{align}
where $C = \max( \norm{ \hat{x}_1 \hat{x}_2 + \hat{p}_1 \hat{p}_2}, \norm{\hat{x}}^2, \norm{\hat{p}}^2)$ and $q = \ceil{\frac{p}{2} - \frac{1}{4}}$. This requires no more than $4 \cdot 5^{q - 1}$ operator exponentials, thus implying
    \begin{align}
        4 \cdot 5^{q - 1} \leq 4 \cdot 5^{ \frac{p}{2} - \frac{1}{4} }.
    \end{align}
\end{proof}

\section{Fock-Space Applications}
We now introduce a series of techniques that allow us to realize polynomials of Fock-space operators. We first begin in \cref{subsec:arbitrary_power} by identifying the error scaling of an arbitrary order Fock-space block encoding, i.e.~$\mathcal{B}_{a^k}$. In \cref{apndx:error-jc}, we show how the techniques can be used to simulate the Jaynes-Cummings Hamiltonian, which itself is a polynomial of Fock-space operators. In \cref{subsec:state-prep}, we demonstrate how this technique can be extended beyond simulation into realizing more general operators, such as a unitary for state preparation.

\subsection{Realizing block encodings of arbitrary order}\label{subsec:arbitrary_power}
We seek to demonstrate the following result:
\begin{theorem}\label{thm:main}
For positive integer $k \geq 1$ and timestep $t \in\mathbb{R}$, we seek to implement the target block encoding $\mathcal{T}_k(t)$ defined as
\begin{align}
    \mathcal{T}_{k}(t) = \exp \left( it \begin{bmatrix}
    0 &  ( a^{\dagger})^{k}  \\
    ( a )^{k} & 0
    \end{bmatrix} \right).
\end{align}
For any $\epsilon>0$ and $p > 1$ there exists an implementable unitary operation $\widetilde{\mathcal{T}}_{k, p}$ of order $p$ such that
\begin{align}
    \norm{\mathcal{T}_k - \widetilde{\mathcal{T}}_{k, p}}\leq \epsilon,
\end{align}
and the number of applications of $\mathcal{S}_1(t)$ needed to implement the operation scales in
\begin{align}
    r \cdot n^{1.6} 30^{np} 420^{n^2 p / 2} 6^{\log_2 n + 1} ,
\end{align}
where $r \in \Theta \left( \frac{(\Lambda^{k/2}t)^{1 + 1/(p - 1)}}{\epsilon^{1 / (p - 1)}} \right)$.
\end{theorem}

Because we can add two lower-order block encodings via \cref{alg:adder}, we can exploit a binary expansion to achieve arbitrary orders (e.g.~$(a^\dagger)^9 = (a^{\dagger})^{2^3} a^\dagger$). Thus, our first task is to demonstrate the implementation of these block encodings with orders that are a power of two. This is achievable through the recursive $\textrm{POWER}$ algorithm:
\begin{algorithm}[H]
\caption{POWER($k, t, p$) }\label{alg:power_of_two}
\begin{algorithmic}
\Require $k = 2^\ell$ for nonnegative integer $\ell$, timestep $t > 0$, order $p > 1$
\Ensure $\widetilde{\mathcal{T}}_k$ with $\norm{\widetilde{\mathcal{T}}_k  - \exp it \begin{bmatrix}
    0 & (a^\dagger)^k \\
    a^k & 0
\end{bmatrix}} \in \mathcal{O}((\Lambda^{k / 2} t)^p)$
\If{$k = 1$}
    \State \Return $\mathcal{S}_1(t)$
\Else
    \State $p' \coloneqq 2p$
    \State $\textrm{HalfOp} \coloneqq \textrm{POWER}(k/2, \sqrt{t/2}, p')$
    \State \Return $\textrm{ADD}(\textrm{HalfOp}, \textrm{HalfOp}, p', p', \sqrt{t / 2})$
\EndIf
\end{algorithmic}
\end{algorithm}

This builds to the following result:
\begin{theorem}\label{lem:key-scaling}
For any $t \geq 0$, $p \geq 1$, and fixed $k = 2^\ell$ for some $\ell \geq 1$ we have that the unitary implemented by~\Cref{alg:power_of_two}, ${\rm POWER}$ acting on $\mathcal{H}_2\otimes \mathcal{H}_\Lambda$, satisfies
\begin{align}
    \norm{ {\rm POWER}(k,t,p) - \exp \left( i t \begin{bmatrix}
    0 & (a^\dagger)^{k} \\
    ( a )^{k}  & 0
    \end{bmatrix} \right) } \in \mathcal{O}((\Lambda^{k/2} t)^{p}),
\end{align}
using no more than $ 6^{\log_2 k} \cdot 420^{kp / 2}$ unitary $\mathcal{S}_1$ operators.
\end{theorem}

To bound the error of this algorithm, we begin by identifying the implementation error of the second-order formula, i.e.~$\widetilde{\mathcal{S}}_2$, the first operator with implementation error:
\begin{lemma}[Implementing second-order block encodings]\label{corr:second_order}
Suppose we can implement the following operation without error (as defined in~\Cref{defn:SX} and subject to a bosonic cutoff)
\begin{align}
\mathcal{S}_1(t) = \exp \left( it 
\begin{bmatrix}
    0 & a^\dagger
    \\
    a & 0
\end{bmatrix} \right).
\end{align}
Then, we can approximate $\mathcal{S}_2(t) $ to the $p^\text{th}$ order, i.e.~implement $\widetilde{\mathcal{S}}_2$ such that
\begin{align}
    \norm{\widetilde{\mathcal{S}}_{2, p}(t) - \mathcal{S}_2(t)} \in \mathcal{O}((\Lambda t)^{p + \frac{1}{2}}),
\end{align}
using no more than $6 \cdot 14^p$ $\mathcal{S}_1(t)$ operations.
\end{lemma}
\begin{proof}
Note that, if $\mathcal{S}_1(t)$ is errorless, then we only need to account for error incurred by the BCH and Trotter formulas. By employing a $p^\text{th}$ order BCH formula we can produce commutator exponentials with error $\mathcal{O}((\Lambda^{1/2} \tau)^{2p + 1})$ by \cref{fact:bch} and \cref{fact:norm-blockencodedxt}.

We again set $\tau = \sqrt{\frac{t}{2}}$ so that the error scales in at worst $\mathcal{O}((\Lambda t)^{p + \frac{1}{2}})$. Then, we apply a Trotter formula \cref{lem:trotter-with-approx} of order $\ceil{\frac{p}{2}}$ so that
\begin{align}
    \norm{\exp \left(i t \begin{bmatrix}
        0 & ( a^\dagger)^2 \\
        ( a )^2 & 0
    \end{bmatrix} \right) - \widetilde{\mathcal{S}}_2(t)} \in \mathcal{O}((Ct)^{p + \frac{1}{2}}),
\end{align}
with $C \leq \max( \Lambda, \Lambda^{2/2}) = \Lambda $. Our worst case error scaling is then $\mathcal{O}((\Lambda t)^{p + \frac{1}{2}})$. This requires no more than $2 \cdot 2 \cdot 5^{\ceil{ \frac{p}{2}} - 1}$ of the commutators, each of which required $8 \cdot 6^{p - 1}$ first order operations. Thus, the cost scales in no more than
\begin{align}
    4 \cdot 5^{p / 2} \cdot 8 \cdot 6^{p - 1} \leq 6 \cdot 14^p
\end{align}
total number of $\mathcal{S}_1$ operations.
\end{proof}

This base case allows us to analyze the performance of \cref{alg:power_of_two}:

\begin{proof}[Proof of \cref{lem:key-scaling}]
We demonstrate the bounds inductively. The base case ($\ell = 1$) holds via \cref{corr:second_order}. For the inductive hypothesis, we assume that, for any $p' \geq 1$ and $k = 2^\ell$, we may implement $\textrm{POWER}(k, t, p')$ as
\begin{align}
    \norm{ {\rm POWER}(k,t, p') - \exp \left( i t \begin{bmatrix}
    0 & ( a^\dagger )^{k} \\
    ( a )^{k}  & 0
    \end{bmatrix} \right) } \in \mathcal{O}((\Lambda^{k / 2} t)^{p'}).
\end{align}
To demonstrate the inductive step, we seek to apply \cref{lem:adder} directly to the implementable operators from the inductive hypothesis. Thus, we set $p' = 2p$ so that
\begin{align}
    \norm{ {\rm POWER}(2k,t, p) - \exp \left( i t \begin{bmatrix}
    0 & (a^\dagger)^{2k} \\
    ( a )^{2k}  & 0
    \end{bmatrix} \right) } \in \mathcal{O}((\Lambda^{2k / 2} t)^{p' / 2}) = \mathcal{O}((\Lambda^{2k/2} t)^{p}),
\end{align}
our desired error scaling. By \cref{lem:adder}, we require an adder of order $\max(\ceil{\frac{p' - 1}{2}}, 1) \leq p + \frac{1}{2}$. Thus, the adder requires $1.07 \cdot 30^{p + 1/2} \leq 6 \cdot 30^p$ of the $\textrm{POWER}(k, t, p')$ operations, namely
\begin{align}
    \textrm{COST}(2k, p) &\leq 6 \cdot 30^p \cdot \textrm{COST}(k, 2p)  \\
    &\leq 6 \cdot 30^p \cdot 6 \cdot 30^{2p} \cdot \textrm{COST}(k / 2, 4p) \\
    &\leq \prod_{j = 1}^{n} 6 \cdot 30^{2^{j - 1}p} \cdot \textrm{COST}(2k / 2^n, 2^n p) \\
    &\leq 6^\ell \cdot 30^{kp} \cdot \textrm{COST}(2, k p).
\end{align}
Since $\textrm{COST}(2, kp) \leq 6 \cdot 14^{kp}$ by \cref{corr:second_order}, the number of $\mathcal{S}_1$ operations is upper bounded by
\begin{align}
    \textrm{COST}(2k, p) \leq  6^{\log_2 k + 1} \cdot 420^{k p} \implies \textrm{COST}(k, p) \leq 6^{\log_2 k} \cdot 420^{kp / 2}.
\end{align}
\end{proof}

Together, the POWER and ADD algorithms allow us to approximate arbitrary orders. We describe a recursive algorithm below to construct any order $k \geq 1$:
\begin{algorithm}[H]
\caption{ARB\_POWER($k, t, p, l, r$) that produces $\mathcal{T}_k(t)$ for any $k \geq 1$}\label{alg:arb_power}
\begin{algorithmic}
\Require $k > 0$ and has the binary representation $k = k_n k_{n-1} ... k_1$.

\If{$r - l = 0$}
    \If{$k_r = 1$}
        \State \Return $\textrm{POWER}(H, 2^r, t, p)$
    \Else
        \State \Return $RX(t) \kron \identity$
    \EndIf
\Else
    \State \Return $\textrm{ADD}(\textrm{ARB\_POWER}(k, \sqrt{t/2}, p, l, \floor{\frac{r - l}{2}} + l), \textrm{ARB\_POWER}(k, \sqrt{t/2}, p, l + \floor{\frac{r - l}{2}} + 1, r) $
\EndIf
\end{algorithmic}
\end{algorithm}

\begin{theorem}\label{thm:arb_calc}
Assuming that the $\mathcal{S}_1$ operator can be implemented without error, the \cref{alg:arb_power} produces a series of gates $\{ \mathcal{S}_1(t_i(t)) \}$ such that
\begin{align}
    \norm{\prod_i \mathcal{S}_1(t_i(t)) - \mathcal{T}_k(t)} \in \mathcal{O}((\Lambda^{k/2} t)^{p}),
\end{align}
where $\mathcal{T}_k$ is our target operator and we have order $k > 0$. The number of $\mathcal{S}_1$ gates required is no more than $n^{1.6} 30^{np} 420^{n^2 p / 2} 6^{\log_2 n + 1}$.
\end{theorem}
\begin{proof}
We demonstrate this constructively on the worst case scenario where $k = k_n k_{n - 1} ... k_1$ and $k_n, k_{n-1}..., k_1 = 1$. Without loss of generality, we assume $n = 2^{\ell}$ for some integer $\ell$. This is because, if $n$ not a power of two, we can simply pad the leading digits with zeros to achieve a balanced binary tree.

The proof is as follows: We first identify the error scaling necessary for each leaf node of the binary tree so that the overall formula has our desired order, and we then perform the cost accounting and estimate the number of $\mathcal{S}_1$ operations required.

To achieve an error scaling of $\mathcal{O}((\Lambda^{2^n / 2} t)^p)$, the two terms being `added' below must have error scaling of order at worst $\mathcal{O}((\Lambda^{2^n / 4} t)^{2p})$ by \cref{lem:adder}, and so on for each subsequent layer. Thus, each of the leaf $\textrm{POWER}$ terms must have order at least $\mathcal{O}((\Lambda^{2^n / 2^{1 + \log_2 n} } t)^{2^{\log_2 n} p }) = \mathcal{O}(( \Lambda^{2^n / 2n} t)^{np})$.

Now, we compute the cost incurred by the formula. We begin by counting the number of times $\textrm{POWER}$ is used, then accounting for the number of $\mathcal{S}_1$ required to implement each $\textrm{POWER}$. Note that each $\textrm{POWER}$ term will be used proportionally to the number of $\textrm{ADD}$ operations necessary, so we can compute the cost of implementing each $\textrm{POWER}$ operation for the $i^\text{th}$ digit, i.e.~the $\textrm{POWER}$ operation has degree $j = 2^i$. We thus bound the number of $\mathcal{S}_1$ operations required to implement this $np^\text{th}$-ordered operator as
\begin{align}
    6^{\log_2 j} \cdot 420^{j n p / 2} = 6^{i} \cdot 420^{2^i n p / 2}.
\end{align}
Finally, we seek to bound the number of times each $\textrm{POWER}$ operator is used through the $\textrm{ADD}$ algorithm. Recall from \cref{lem:adder} that each $\textrm{ADD}$ operation requires at most $1.07 \cdot 30^q$ of the constituent operators, where $q = \max ( \ceil{\frac{\min(p_l, p_r) - 1}{2}}, 1)$ and $p_l, p_r$ are the orders of the underlying operators. By assuming symmetry of the Trotter formula for $\textrm{ADD}$, each addition requires $\frac{1}{2} 1.07 \cdot 30^q$ of the underlying operator. Thus, we can obtain a bound on the number of applications required of each fundamental $\textrm{POWER}$ operator as
\begin{align}
    \prod_{s = 1}^{\log_2 n} \frac{1}{2} 1.07 \cdot 30^{ 2^{s - 1}p + 1/2} \leq 3^{\log_2 n} 30^{np} \leq n^{1.6} 30^{np}
\end{align}
because the $s^\text{th}$ layer of $\textrm{ADD}$ requires constituent operators of order $2^s p$, so $q \leq 2^{s - 1}p + \frac{1}{2}$. 

Finally, considering the total cost by adding up the cost of the individual $\textrm{POWER}$ operators multiplied by the number of applications required yields
\begin{align}
    \sum_{i = 1}^{\log_2 n} n^{1.6} 30^{np} \cdot 6^{i} \cdot 420^{2^{i} np / 2} &\leq n^{1.6} 30^{np} 420^{n^2 p / 2} \sum_{i = 1}^{\log_2 n} 6^i \\
    &\leq n^{1.6} 30^{np} 420^{n^2 p / 2} 6^{\log_2 n + 1},
\end{align}
as desired.
\end{proof}

Now, we seek to finalize the number of operations required in terms of $\epsilon$. Recognize that we may use timeslicing to reduce the error arbitrarily. Note that:
\begin{lemma}\label{r-scaling}
Suppose we may implement $\widetilde{\mathcal{T}}_{k, p}(t)$, an approximation of $\mathcal{T}_k(t)$ with $p > 1$ such that
\begin{align}
    \norm{\mathcal{T}_k(t) - \widetilde{\mathcal{T}}_{k, p}(t)} \in \mathcal{O}((\Lambda^{k / 2} t)^p).
\end{align}
Then, by timeslicing the approximation, we can produce $\widetilde{\mathcal{T}}_{k, p}^r(t)$ where
\begin{align}
    \norm{\mathcal{T}_k(t) - \widetilde{\mathcal{T}}_{k, p}^r(t)} \leq \epsilon,
\end{align}
where $\widetilde{\mathcal{T}}_k^r(t)$ requires $r \in \Theta \left( \frac{(\Lambda^{k/2}t)^{1 + 1/(p - 1)}}{\epsilon^{1 / (p - 1)}} \right)$ applications of the $\widetilde{\mathcal{T}}_k(t)$ operator.
\end{lemma}
\begin{proof}
We define the timeslicing of $\widetilde{T}_k(t)$ as applying $\widetilde{\mathcal{T}}_k(t/r)$ operator $r$ times
\begin{align}
    \widetilde{\mathcal{T}}_{k, p}^r(t) = \widetilde{\mathcal{T}}_{k, p}(t / r)^r.
\end{align}
To find a Taylor expansion for $\widetilde{\mathcal{T}}_{k, p}(t / r)^r$, note the explicit form for $\widetilde{\mathcal{T}}_{k, p}(t)$
\begin{align}
    \widetilde{\mathcal{T}}_{k, p}(t) = e^{i A_k t} + \Delta(t) (\Lambda^{k /2} t)^p,
\end{align}
where $\norm{\Delta(t)} \in \mathcal{O}(1)$. This allows us to express the refined operator as
\begin{align}
    \widetilde{\mathcal{T}}_k(t/r)^r &= \left( e^{i A_k t / r} + \Delta\left( \frac{t}{r}\right) \left( \frac{t^{p}}{r^{p}} \right) \right)^r \\
    &= e^{i A_k t} + \left[ \sum_{j = 1}^{r - 1} (e^{i A_k t / r})^{j} \Delta\left( \frac{t}{r} \right) (e^{i A_k t / r})^{r - 1 - j} \right] \left( \frac{(\Lambda^{k/2}t)^{p}}{r^{p}} \right) + \mathcal{O} \left(\left( \frac{\Lambda^{k/2} t}{r} \right)^{p + 1} \right),
\end{align}
Thus, when we analyze the error,
\begin{align}
    \norm{\mathcal{T}_k(t/r)^r - e^{i A_k t}} &\leq \mathcal{O} \left( r \norm{(e^{i A_k t /r})^{r - 1}} \norm{\Delta\left( \frac{t}{r} \right)} \left(  \frac{(\Lambda^{k/2}t)^{p}}{r^{p}} \right) \right) \nonumber\\
    &\subset \mathcal{O} \left( \frac{(\Lambda^{k/2}t)^{p}}{r^{p - 1}} \right).
\end{align}
To bound the implementation error by $\epsilon$, i.e.~$\epsilon \in \mathcal{O} \left(  \frac{(\Lambda^{k/2}t)^{p}}{r^{p - 1}} \right)$, we should select $r$ as
\begin{align}
    r \in \Theta \left( \frac{(\Lambda^{k/2}t)^{p / (p - 1)}}{\epsilon^{1 / (p - 1)}} \right) = \Theta \left( \frac{(\Lambda^{k/2}t)^{1 + 1/(p - 1)}}{\epsilon^{1 / (p - 1)}} \right),
\end{align}
as desired.
\end{proof}

Finally, we demonstrate our original theorem statement, which allows us to create a bound on the number of $\mathcal{S}_1$ operations necessary to achieve an arbitrarily ordered operator:
\begin{proof}[Proof of \cref{thm:main}]
By \cref{thm:arb_calc}, we can perform a single Trotter step of timestep $\frac{t}{r}$ using $n^{1.6} 30^{np} 420^{n^2 p / 2} 6^{\log_2 n + 1}$ $\mathcal{S}_1$ operations. Thus, the total number of $\mathcal{S}_1$ operations required scales in
\begin{align}
    r \cdot n^{1.6} 30^{np} 420^{n^2 p / 2} 6^{\log_2 n + 1} ,
\end{align}
where, by \cref{r-scaling}, it is sufficient to set $r \in \Theta \left( \frac{(\Lambda^{k/2}t)^{1 + 1/(p - 1)}}{\epsilon^{1 / (p - 1)}} \right)$. 

\end{proof}

\subsection{Generation of nonlinear Hamiltonians}\label{apndx:error-jc}

\resultjc

\begin{proof}
We first show that the two Hamiltonian terms are implementable separately. Then, via Trotter, we combine them and perform an error analysis. In particular, we hope to embed the Hamiltonian such that we approximate the  operator
\begin{align}
    \exp it \begin{bmatrix}
        H & 0 \\
        0 & \cdot
    \end{bmatrix},
\end{align}
i.e., where the Hamiltonian is embedded in the upper left hand block. Thus, when applied to a system with the $\ket{0}$ qubit, this amounts to implementing $\exp i t H $ on the mode.  

We begin by embedding the $a^\dagger a$ term. Notice that $a^\dagger a$ Hermitian; thus, we can apply \cref{alg:mult} on the $a, a^\dagger$ block encodings. By the error analysis in \cref{fact:bch} and the bound on the norm from \cref{fact:norm-blockencodedxt},
\begin{align}
    \norm{\textrm{MULT} \left(i \tau \mathcal{B}_{a^\dagger}, i\tau \mathcal{B}_{a} \right) - \exp 2 i \tau^2 \begin{bmatrix}
         a^\dagger   a  & 0 \\
        0 &  - a  a^\dagger  
    \end{bmatrix} } \in \mathcal{O}((\Lambda^{1/2} \tau)^{2q + 1}).
\end{align}
By setting $\tau = \sqrt{\frac{\omega t}{2}}$, we can achieve the $a^\dagger a$ term:
\begin{align}
    \norm{\textrm{MULT} \left(i \sqrt{\frac{\omega t}{2}} \mathcal{B}_{a^\dagger}, i\sqrt{\frac{\omega t}{2}} \mathcal{B}_{a} \right) - \exp i t \omega \begin{bmatrix}
         a^\dagger   a  & 0 \\
        0 &  - a  a^\dagger  
    \end{bmatrix} },
\end{align}
requiring $8 \cdot 6^{q - 1}$ total $\exp i t \mathcal{B}_{a^\dagger}$ operations (which can be reduced to $\mathcal{S}_1$ operations via conjugation). 

We now tackle the second order term. Recall that we can block-encode $( a^\dagger )^2$ and $ ( a )^2$ via \cref{alg:power_of_two}. We then apply \cref{alg:mult} to $(a^\dagger)^2, a^2$ to yield the desired upper-left block encoding. Namely, we can implement $\exp i t \mathcal{B}_{(a^\dagger)^2}$ with the following error scaling:
\begin{align}
    \norm{\exp it \widetilde{\mathcal{B}}_{(a^\dagger)^2} - \exp i \tau \mathcal{B}_{(a^\dagger)^2}} \in \mathcal{O}((\Lambda \tau)^{p + \frac{1}{2}}),
\end{align}
using no more than $6 \cdot 14^p$ $\mathcal{S}_1$ operations, we can apply \cref{alg:mult} to find
\begin{align}
    \norm{\textrm{MULT} (i\tau \mathcal{B}_{(a^\dagger)^2}, i \tau \mathcal{B}_{(a)^2} ) - \exp 2 i \tau^2 \begin{bmatrix}
        ( a^\dagger )^2 ( a )^2 & 0 \\
        0 & ( a )^2 ( a^\dagger )^2
    \end{bmatrix} } \in \mathcal{O}((\Lambda^2 \tau)^{p + \frac{1}{2}}),
\end{align}
by setting $\ell = \ceil{\frac{p - \frac{1}{2}}{2}} \leq \frac{p}{2} + 1 $ and thus using $8 \cdot 6^{\ell - 1}$ exponentials. When $\tau = \sqrt{\frac{\kappa t}{4}}$,
\begin{align}
    \norm{\textrm{MULT} \left(i \sqrt{\frac{\kappa t}{4}} \mathcal{B}_{(a^\dagger)^2}, i \sqrt{\frac{\kappa t}{4}} \mathcal{B}_{(a)^2} \right) - \exp i t \frac{\kappa}{2} \begin{bmatrix}
        ( a^\dagger )^2 ( a )^2 & 0 \\
        0 & ( a )^2 ( a^\dagger )^2
    \end{bmatrix} } \in \mathcal{O}( (\Lambda^4 \kappa t)^{\frac{p}{2} + \frac{1}{4}}),
\end{align}
using no more than $8 \cdot 6^{\ell - 1} \cdot 6 \cdot 14^p \leq 48 \cdot 6^{p / 2} \cdot 14^{p} \leq 48 \cdot 35^p$ total $\mathcal{S}_1$ operators. Then, we may set $p = 2q$ so that, given no more than $48 \cdot 35^{2q}$ total $\mathcal{S}_1$ operations, we can implement the BCH formula with error scaling $\mathcal{O}((\Lambda^4 \kappa t)^{q + \frac{1}{4}})$.

Finally, we apply the Trotter formula to the two subterms via \cref{lem:trotter-with-approx}. Define the approximate matrix exponentials as follows:
\begin{align}
    \textrm{FIRST} &\coloneqq \log \textrm{MULT} \left(i \sqrt{\frac{\omega t}{2}} \mathcal{B}_{a^\dagger}, i\sqrt{\frac{\omega t}{2}} \mathcal{B}_{a} \right) \\
    \textrm{SECOND} &\coloneqq \log \textrm{MULT} \left(i \sqrt{\frac{\kappa t}{4}} \mathcal{B}_{(a^\dagger)^2}, i \sqrt{\frac{\kappa t}{4}} \mathcal{B}_{(a)^2} \right), 
\end{align}
so that,
\begin{align}
    \norm{\trotter_{2s}\left(\textrm{FIRST}, \textrm{SECOND}\right)  - \exp i t \begin{bmatrix}
        H & 0 \\
        0 & \cdot
    \end{bmatrix} } \in \mathcal{O}((\Lambda^4 \max( \omega, \kappa) t)^{q + \frac{1}{4}})
\end{align}
where the constant factor can be obtained by observing the Hamiltonian norm is bounded via the triangle inequality. Now, by setting $s = \ceil{\frac{1}{2} (q - \frac{3}{4})} \leq \frac{q}{2} + \frac{5}{8}$, we can obtain the desired error scaling. This formula requires no more than $4 \cdot 5^{s - 1} $ total operations; by symmetry, we can assume each of the BCH formulas only must be applied $2 \cdot 5^{s - 1}$ times. Thus, the total number of $\mathcal{S}_1$ operations is no more than
\begin{align}
    2 \cdot 5^{s - 1} (8 \cdot 6^{q - 1} + 48 \cdot 35^{2q}) \leq 4 \cdot 5^{\frac{q}{2}} \cdot 48 \cdot 35^{2q} \leq 192 \cdot 2900^q.
\end{align}

To produce an $\epsilon$ scaling, we apply \cref{r-scaling} to the Trotterized operator, implying that we require the following $r$ scaling for fixed $q$:
\begin{align}
    r \in \Omega\left( \frac{(\Lambda^{4} t)^{1 + 1 / (q - \frac{3}{4})} }{\epsilon^{1 / (q - \frac{3}{4})}} \right),
\end{align}
where the total number of $\mathcal{S}_1$ operations is no more than
\begin{align}
    r \cdot 192 \cdot 2900^q \subset re^{\mathcal{O}(q)}.
\end{align}

\end{proof}

\subsection{Application to state preparation}\label{subsec:state-prep}
We first need to demonstrate the connection between block-encoded powers of annihilation/creation operators and state preparation. First, observe that the ideal block encoding would allow for initialization from the vacuum:
\statepreptime
\begin{proof}\label{state_prep_proof}
    The proof is algebraic; begin by producing the Taylor series expansion of the operator:
    \begin{align}
        \mathcal{T}_{k}(t) &= \exp \left( it \begin{bmatrix}
            0 & ( a^\dagger )^k \\
            ( a )^k & 0
        \end{bmatrix} \right) \\
        &\qquad= \sum_{j = 0}^\infty \frac{(it)^j}{j!} \begin{bmatrix}
            0 & ( a^\dagger )^k \\
            ( a )^k & 0
        \end{bmatrix}^j \\ 
        &\qquad= \sum_{j = 0}^\infty \frac{(it)^{2j}}{(2j)!} \begin{bmatrix}
            0 & ( a^\dagger )^k \\
            ( a )^k & 0
        \end{bmatrix}^{2j}
        + \sum_{j = 0}^\infty \frac{(it)^{2j + 1}}{(2j + 1)!} \begin{bmatrix}
            0 & ( a^\dagger )^k \\
            ( a )^k & 0
        \end{bmatrix}^{2j + 1},
    \end{align}
where the matrix products have well-defined forms
    \begin{align}
        \begin{bmatrix}
            0 & ( a^\dagger )^k \\
            ( a )^k & 0
        \end{bmatrix}^{2j} &= \begin{bmatrix}
            (( a^\dagger )^k ( a )^k)^j & 0 \\
            0 & (( a )^k ( a^\dagger )^k)^j 
        \end{bmatrix}, \\
        \begin{bmatrix}
            0 & ( a^\dagger )^k \\
            ( a )^k & 0
        \end{bmatrix}^{2j + 1} &= \begin{bmatrix}
             0 & ( a^\dagger )^k (( a )^k ( a^\dagger )^k)^j \\
            ( a )^k (( a^\dagger )^k ( a )^k)^j & 0 
        \end{bmatrix},
    \end{align}
so that
\begin{align}
    &\exp \left( it \begin{bmatrix}
            0 & ( a^\dagger )^k \\
            ( a )^k & 0
        \end{bmatrix} \right) \ket{1} \kron \ket{0} \\
        &\qquad = \sum_{j = 0}^\infty \frac{(it)^{2j}}{(2j)!} \sqrt{k!}^{2j} \ket{1} \kron \ket{0} + \sum_{j = 0}^\infty \frac{(it)^{2j + 1}}{(2j + 1)!} \sqrt{k!}^{2j + 1}  \ket{0} \kron \ket{k} \\
        &\qquad= \cos (t \sqrt{k!}) \ket{1} \kron \ket{0} + i \sin (t \sqrt{k!}) \ket{0} \kron \ket{k}.
\end{align}
When $t \sqrt{k!} = (2n + 1) \frac{\pi}{2}$ for $n \in \mathbb{N}$, the $\ket{1} \kron \ket{0}$ term vanishes and we are left with the $\ket{0} \kron \ket{k}$ Fock state, as desired.
\end{proof}

While this result allows us to prepare the $\ket{k}$ Fock state, it also will incur unwanted transformations on starting states other than the vacuum ($\ket{1} \kron \ket{n}, n \neq 1$). By applying the BCH formula, we can isolate this operation so that it only operates on the $\ket{1} \kron \ket{0}$ term. In particular, we argue:
\resultstateprep
\begin{proof}
The general construction of the operator emerges from the use of a Trotter formula in conjunction with a phase rotation gate. We begin by defining the rotation operator:
\begin{define}
Call $R_{Z0}$ the phase-flip operator acting on some set of modes $B$ to be
\begin{align}
    R_{Z0} := \identity \otimes (\identity - 2 \ket{0}\!\bra{0}),
\end{align}
i.e., only flip the phase for the vacuum. This operator is implementable using a 0-controlled cavity-conditioned qubit rotation gate. 

\end{define}
Then, because $R_{Z0}$ is self-adjoint, we can conjugate $\mathcal{T}_k(-t)$ without error as 
\begin{align}
    R_{Z0} \mathcal{T}_k (-t) R_{Z0} &= R_{Z0} \exp \left( it' \begin{bmatrix}
            0 & ( a^\dagger )^k \\
            ( a )^k & 0
        \end{bmatrix} \right) R_{Z0} \\
        &= \exp \left( - it R_{Z0} \begin{bmatrix}
            0 & ( a^\dagger )^k \\
            ( a )^k & 0
        \end{bmatrix} R_{Z0} \right) \\
        &= \exp \left( it  \begin{bmatrix}
            0 & ( a^\dagger )^k (2 \ket{0}\!\bra{0} - \identity) \\
            (2 \ket{0}\!\bra{0} - \identity) ( a )^k  & 0
        \end{bmatrix} \right),
\end{align}
where specific left- and right-hand $(2 \ket{0}\!\bra{0} - \identity)$ terms vanish given the annihilation/creation operators. We then apply the Trotter formula upon $\mathcal{T}_k(t/2), R_{Z0} \mathcal{T}_k(-t/2) R_{Z0}$, which yields
\begin{align}
    \exp \left( i t \begin{bmatrix}
        0 & ( a^\dagger )^k \ketbra{0}{0} \\
        \ketbra{0}{0} ( a )^k & 0
    \end{bmatrix} \right),
\end{align}
as desired.

To compute the error scaling, recall our result from \cref{thm:arb_calc}, which states the error scaling of $\widetilde{\mathcal{T}}_{k, p}$ (and, respectively, $R_{Z0} \widetilde{\mathcal{T}}_{k, p} R_{Z0}$) is
\begin{align}
    \norm{\widetilde{\mathcal{T}}_{k, p}(t) - \mathcal{T}_k (t)} \in \mathcal{O}((\Lambda^{k/2} t)^p) ,
\end{align}
so that, by $\cref{lem:trotter-with-approx}$, 
\begin{align}
    \norm{\trotter_{2q}(\widetilde{\mathcal{T}}_{k, p}(t/2), R_{Z0} \widetilde{\mathcal{T}}_{k, p}(-t/2) R_{Z0}) - \exp \left( i t \begin{bmatrix}
        0 & ( a^\dagger )^k \ketbra{0}{0} \\
        \ketbra{0}{0} ( a )^k & 0
    \end{bmatrix} \right) } \in \mathcal{O}((\Lambda^{k/2} t)^{p}),
\end{align}
when $q = \max(\ceil{\frac{p - 1}{2}}, 1 ) $ and using no more than $4 \cdot 5^{q-1}$ operator exponentials. Thus, we can set $q = \frac{p + 1}{2} \geq \max(\ceil{\frac{p - 1}{2}}, 1 ) $ so that we use no more than $ 4 \cdot 5^{\frac{p - 1}{2}} \leq 2 \cdot 5^{p/2}$ $\widetilde{\mathcal{T}}_{k, p}$ terms.

\end{proof}

Thus, our approximate operators can be applied to yield the same result with high probability:
\begin{theorem}
We can prepare the $\ket{0} \kron \ket{k}$ with probability at least $ 1- \delta$ using no more than $r \in \Theta \left( \frac{(\Lambda^{k/2}t)^{1 + 1/(p - 1)}}{(\delta / 2)^{1 / (p - 1)}} \right)$ $\mathcal{F}_{k, p}$ operators or at most
\begin{align}
    r \cdot 2 \cdot 5^{p/2} \cdot n^{1.6} 30^{np} 420^{n^2 p / 2} 6^{\log_2 n + 1} 
\end{align}
$\mathcal{S}_1$ operators.

\end{theorem}

\begin{proof}
Begin by identifying the $\epsilon$ precision necessary to yield a failure probability less than $\delta$. A sufficient condition would be that
\begin{align}
    \left| \norm{\ketbra{0}{0} \kron \ketbra{k}{k} \stateprepapprox \ket{1} \kron \ket{0}}^2 - \norm{\ketbra{0}{0} \kron \ketbra{k}{k} \stateprep \ket{1} \kron \ket{0}}^2 \right| \leq \delta.
\end{align}
Observe that our idealized operator has a success probability; thus, we seek to demonstrate that
\begin{align}
    \left|\norm{\ketbra{0}{0} \kron \ketbra{k}{k} \stateprepapprox \ket{1} \kron \ket{0}}^2 - 1\right| \leq \delta .
\end{align}
Because the probability of measuring $\ket{0} \kron \ket{k}$ lies in $[0, 1]$, the above inequality holds when
\begin{align}
    \norm{\ketbra{0}{0} \kron \ketbra{k}{k} \stateprepapprox \ket{1} \kron \ket{0}}^2 \geq 1 - \delta.
\end{align}
We recognize that we can lower bound the norm
\begin{align}
    &\norm{\ketbra{0}{0} \kron \ketbra{k}{k} \stateprepapprox \ket{1} \kron \ket{0}}\\
    &\qquad= \norm{\ketbra{0}{0} \kron \ketbra{k}{k} (\stateprep - (\stateprep - \stateprepapprox)) \ket{1} \kron \ket{0}} \\
    &\qquad\geq \left| \norm{\ketbra{0}{0} \kron \ketbra{k}{k} \mathcal{F}_{k} \ket{1} \kron \ket{0}} - \norm{\ketbra{0}{0} \kron \ketbra{k}{k} (\stateprep - \stateprepapprox) \ket{1} \kron \ket{0}} \right| \\
    &\qquad\geq 1 - \norm{\stateprep - \stateprepapprox}_\infty
\end{align}
by requiring $\norm{\stateprep - \stateprepapprox}_\infty \leq 1$. This allows us to produce a lower bound on the original LHS
\begin{align}
    \norm{\ketbra{0}{0} \kron \ketbra{k}{k} \stateprepapprox \ket{1} \kron \ket{0}}^2 \geq 1 - 2 \norm{\stateprep - \stateprepapprox}_\infty.
\end{align}
Thus, it is sufficient to hold
\begin{align}
    1 - 2 \norm{\stateprep - \stateprepapprox}_\infty \geq 1 - \delta \iff \norm{\stateprep - \stateprepapprox}_\infty \leq \frac{\delta}{2}.
\end{align}
We apply \cref{r-scaling} to \cref{lem:fock-prep-unitary} so that the time-sliced $\stateprepapprox^r$ has
\begin{align}
    \norm{\stateprepapprox^r - \stateprep} \leq \frac{\delta}{2},
\end{align}
by using $r \in \Theta \left( \frac{(\Lambda^{k/2}t)^{1 + 1/(p - 1)}}{(\delta / 2)^{1 / (p - 1)}} \right)$ applications of $\stateprepapprox(t/r)$. The $\mathcal{S}_1$ bound follows from a similar analysis to \cref{thm:main} applied to the result from \cref{lem:fock-prep-unitary}.

\end{proof}

\section{Universal Control of the \texorpdfstring{Span $\left\{ \left|0\right\rangle ,\left|1\right\rangle \right\} $}{}
Fock Space\label{subsec:Universal-Control}}

To further demonstrate the efficacy of the instruction set, we employ the approach to encode a qubit in a cavity either via generation
of effective Pauli gates 
or imposition of an effective Hubbard interaction in the Jaynes-Cumming
Hamiltonian.  In this sense, the techniques presented here are analogous to those in~\cite{liu2021constructing}, in that we use our results to effectively truncate the quantum information to a two-dimensional subspace despite the fact that the natural dynamics of the systems causes the quantum information to leak from this space into the larger Hilbert space of the cavity. Error bounds for these complex operations are considered outside of the scope of the present work and are left for future research.

For universal control in the restricted $\text{span}\left\{ \left|0\right\rangle ,\left|1\right\rangle \right\} $
Hilbert space, we generate three effective Pauli operators $\sigma_{\text{eff}}^{x}$,
$\sigma_{\text{eff}}^{y}$, and $\sigma_{\text{eff}}^{z}$ that produce
Pauli rotations in the lowest two modes of the cavity, with minimal leakage
to higher energy states. The form of the effective Pauli operators
is determined by expressing the standard Pauli operators
\begin{align}
\sigma^{x} & =\left(\begin{array}{cc}
0 & 1\\
1 & 0
\end{array}\right),\\
\sigma^{y} & =\left(\begin{array}{cc}
0 & -i\\
i & 0
\end{array}\right),\\
\sigma^{z} & =\left(\begin{array}{cc}
1 & 0\\
0 & -1
\end{array}\right),
\end{align}
in terms of creation and annihilation operators truncated to the first
two Fock states
\begin{align}
\hat{a}_{\text{eff}}^{\dagger} & =\left(\begin{array}{cc}
0 & 0\\
1 & 0
\end{array}\right),\\
\hat{a}_{\text{eff}} & =\left(\begin{array}{cc}
0 & 1\\
0 & 0
\end{array}\right),\\
\hat{n}_{\text{eff}} & =\hat{a}^{\dagger}\hat{a}_{\text{eff}}=\left(\begin{array}{cc}
0 & 0\\
0 & 1
\end{array}\right),
\end{align}
which yields
\begin{align}
\sigma_{\text{eff}}^{x} & =\hat{a}_{\text{eff}}^{\dagger}+\hat{a}_{\text{eff}},\\
\sigma_{\text{eff}}^{y} & =i\left(\hat{a}_{\text{eff}}^{\dagger}-\hat{a}_{\text{eff}}\right),\\
\sigma_{\text{eff}}^{z} & =I-2\hat{a}_{\text{eff}}^{\dagger}\hat{a}_{\text{eff}}.
\end{align}
To reduce leakage into higher energy states, we ensure the creation
operator $\hat{a}_{\text{eff}}^{\dagger}$ only acts on the ground
state $\left|0\right>$ and the annihilation operator $\hat{a}_{\text{eff}}$
only acts on the first excited state $\left|1\right>$ with the projector
\begin{align}
\hat{P}_{0} & \approx I-\hat{n}\\
 & =\begin{cases}
0 & n=1\\
1 & n=0
\end{cases},
\end{align}
where $n$ is the number of photons in the cavity and where only the
$\text{span}\left\{ \left|0\right\rangle ,\left|1\right\rangle \right\} $
states are populated. Since the operator is a projector, it obeys
the relation
\begin{equation}
\hat{P}_{0}^{2}=\hat{P}_{0},
\end{equation}
such that the effective Pauli gates are
\begin{align}
\sigma_{\text{eff}}^{x} & =\hat{a}_{\text{eff}}^{\dagger}\hat{P}_{0}+\hat{P}_{0}\hat{a}_{\text{eff}}\\
 & \approx\hat{a}^{\dagger}\left(I-\hat{n}\right)+\left(I-\hat{n}\right)\hat{a},\\
\sigma_{\text{eff}}^{y} & =i\left(\hat{a}_{\text{eff}}^{\dagger}\hat{P}_{0}-\hat{P}_{0}\hat{a}_{\text{eff}}\right)\\
 & \approx i\left(\hat{a}^{\dagger}\left(I-\hat{n}\right)-\left(I-\hat{n}\right)\hat{a}\right),\\
\sigma_{\text{eff}}^{z} & =I-2\hat{a}_{\text{eff}}^{\dagger}\hat{P}_{0}^{2}\hat{a}_{\text{eff}}\\
 & \approx I-2\hat{a}^{\dagger}\left(I-\hat{n}\right)\hat{a}.
\end{align}

\subsubsection*{Pauli X Gate}

Consider the infinitesimal $\sigma_{x}$-rotation gate in the $\text{span}\left\{ \left|0\right\rangle ,\left|1\right\rangle \right\} $
Fock space 
\begin{align}
U_{\text{span}\left\{ 0,1\right\} ,x} & =e^{i\lambda^{2}\sigma_{\text{eff}}^{x}\sigma^{z}}\\
 & =e^{i\lambda^{2}\left(\hat{a}^{\dagger}\left(1-\hat{n}\right)+\left(1-\hat{n}\right)\hat{a}\right)\sigma^{z}}.
\end{align}
Expression of the exponent in terms of phase-space operators Eqs.~\ref{eq:AnnihilationOperatortoPhaseSpace}, \ref{eq:CreationOperatortoPhaseSpace}, and \ref{eq:NumbertoPhaseSpace}
gives 
\begin{gather}
i\lambda^{2}\left(\hat{a}^{\dagger}\left(I-\hat{n}\right)+\left(I-\hat{n}\right)\hat{a}\right)\sigma^{z}\nonumber \\
=i\lambda^{2}\left(2\hat{x}-\left\{ \hat{x},\hat{n}\right\} +i\left[\hat{p},\hat{n}\right]\right)\sigma^{z}.
\end{gather}
The gate is therefore given by Trotter for each of
three terms: $\exp\left(\left[\hat{A}_{1},\hat{B}_{1}\right]\lambda^{2}\right)=\exp\left(-i\lambda^{2}\left\{ \hat{x},\hat{n}\right\} \sigma^{z}\right)$,
$\exp\left(\left[\hat{A}_{2},\hat{B}_{2}\right]\lambda^{2}\right)=\exp\left(-\lambda^{2}\left[\hat{p},\hat{n}\right]\sigma^{z}\right)$,
and $\exp\left(2i\lambda^{2}\hat{x}\sigma^{z}\right)$. 

The terms consisting of exponentials of commutators are decomposed via BCH. The relationship between the commutator and anticommutator required for the first term is given by the Pauli anticommutation-commutation relation
\begin{align}
-i\left\{ \hat{x},\hat{n}\right\} \sigma^{z} & =-i\left(i\left[i\hat{x}\sigma^{x},i\hat{n}\sigma^{y}\right]\right)\\
 & =\left[i\hat{x}\sigma^{x},i\hat{n}\sigma^{y}\right]\\
 & =\left[\hat{A}_{1},\hat{B}_{1}\right],
\end{align}
where $\hat{A}_{1}$ corresponds to a position displacement and $\hat{B}_{1}$
corresponds to the $y$-conditional rotation gate. The argument of
the second term is already in the form of a commutator, such that
\begin{align}
\left[\hat{A}_{2},\hat{B}_{2}\right] & =-\left[\hat{p},\hat{n}\right]\sigma^{z}\\
 & =\left[i\hat{p},i\hat{n}\sigma^{z}\right],
\end{align}
where $\hat{A}_{2}$ corresponds to an \emph{unconditional} momentum boost,
and $\hat{B}_{2}$ corresponds to the $z$-conditional rotation gate.
Lastly, the third term already belongs to the instruction set architecture
and needs no further decomposition. 

The infinitesimal $\sigma_{x}$-rotation gate in the $\text{span}\left\{ \left|0\right\rangle ,\left|1\right\rangle \right\} $
Fock space is therefore composed of a product of nine rotation and displacement gates or 21 displacement gates.

\subsubsection*{Pauli Y Gate}

The infinitesimal $\sigma_{y}$-rotation gate in the $\text{span}\left\{ \left|0\right\rangle ,\left|1\right\rangle \right\} $
Fock space is determined analogously 
\begin{align}
U_{\text{span}\left\{ 0,1\right\} ,y} & =e^{i\lambda^{2}\sigma_{\text{eff}}^{y}\sigma^{z}}\\
 & =e^{-\lambda^{2}\left(\hat{a}^{\dagger}\left(I-\hat{n}\right)+\left(I-\hat{n}\right)\hat{a}\right)\sigma^{z}}.
\end{align}
Expression of the argument of the exponent in terms of phase-space
variables Eq.~\ref{eq:CreationOperatortoPhaseSpace} and Eq.~\ref{eq:AnnihilationOperatortoPhaseSpace}
yields 
\begin{gather}
-\lambda^{2}\left(\hat{a}^{\dagger}\left(I-\hat{n}\right)-\left(I-\hat{n}\right)\hat{a}\right)\sigma^{z}\nonumber \\
=-\lambda^{2}\left(-2i\hat{p}+\left[\hat{n},\hat{x}\right]+i\left\{ \hat{n},\hat{p}\right\} \right)\sigma^{z},
\end{gather}
such that the gate is a Trotter-Suzuki decomposition 
of $\exp\left(\left[\hat{A}_{1},\hat{B}_{1}\right]\lambda^{2}\right)=\exp\left(-\lambda^{2}\left[n,x\right]\sigma^{z}\right)$,
$\exp\left(\left[\hat{A}_{2},\hat{B}_{2}\lambda^{2}\right]\right)=\exp\left(-i\lambda^{2}\left\{ p,n\right\} \sigma^{z}\right)$,
and $\exp\left(2i\lambda^{2}p\sigma^{z}\right)$. 

Again, the first two exponential terms are decomposed via BCH. The first commutator is 
\begin{align}
\left[\hat{A}_{1},\hat{B}_{1}\right] & =\left[\hat{n},\hat{x}\right]\sigma^{z}\\
 & =\left[\hat{n}\sigma^{z},\hat{x}\right],
\end{align}
where the exponent of $\hat{A}_{1}$ is a $z$-conditional rotation
gate and the exponent of $\hat{B}_{1}$ is an unconditional position
displacement. The second commutator is given by the Pauli anticommutation-commutation
relation
\begin{align}
i\left\{ p,n\right\} \sigma^{z} & =i\left(i\left[ip\sigma^{x},in\sigma^{y}\right]\right)\\
 & =\left[ip\sigma^{x},in\sigma^{y}\right]\\
 & =\left[\hat{A}_{2},\hat{B}_{2}\right],
\end{align}
where the exponent of $\hat{A}_{2}$ corresponds to a conditional
momentum shift and the exponent of $\hat{B}_{2}$ is a $y$-conditional
rotation gate.

The infinitesimal $\sigma_{y}$-rotation gate in the $\text{span}\left\{ \left|0\right\rangle ,\left|1\right\rangle \right\} $
Fock space therefore has a lower bound gate depth of nine displacement and rotation gates or 21 in displacement gates.

\subsubsection*{Pauli Z Gate}

The infinitesimal $\sigma_{z}$-rotation gate in the $\text{span}\left\{ \left|0\right\rangle ,\left|1\right\rangle \right\} $
Fock space is 
\begin{align}
U_{\text{span}\left\{ 0,1\right\} ,z} & =e^{i\lambda^{2}\sigma_{\text{eff}}^{z}\sigma^{z}}\\
 & =e^{-\lambda^{2}\left(I-2\hat{a}^{\dagger}\left(I-\hat{n}\right)\hat{a}\right)\sigma^{z}},
\end{align}
whose argument in terms of ladder operators is
\begin{gather}
-\lambda^{2}\left(I-2\hat{a}^{\dagger}\left(I-\hat{n}\right)\hat{a}\right)\sigma^{z}\nonumber \\
=-\lambda^{2}\left(I-2\hat{a}^{\dagger}a+2\hat{a}^{\dagger}\hat{a}^{\dagger}\hat{a}\hat{a}\right)\sigma^{z},
\end{gather}
Given the ladder operator commutator Eq.~\ref{eq:CommutatorCreationAnnihilation},
\begin{equation}
\hat{a}^{\dagger}\hat{a}=\hat{a}\hat{a}^{\dagger}-I,
\end{equation}
the relationship between the fourth-order ladder operator term and
the number operator is
\begin{align}
\hat{a}^{\dagger}\hat{a}^{\dagger}\hat{a}\hat{a} & =\hat{a}^{\dagger}\left(\hat{a}\hat{a}^{\dagger}-I\right)\hat{a}\\
 & =\hat{a}^{\dagger}\hat{a}\hat{a}^{\dagger}\hat{a}-\hat{a}^{\dagger}\hat{a}\\
 & =\hat{n}^{2}-\hat{n}.
\end{align}
The argument of the exponential in terms of number operators is then
\begin{equation}
-\lambda^{2}\left(I-2\hat{a}^{\dagger}\left(I-\hat{n}\right)\hat{a}\right)\sigma^{z}=-\lambda^{2}\left(I-4\hat{n}+2\hat{n}^{2}\right)\sigma^{z}.
\end{equation}
The argument is further simplified given that the state is restricted
to the first two cavity modes, as for $n=0$ and $n=1$ the quantity
$\hat{n}^{2}-\hat{n}$ is zero, such that
\begin{gather}
-\lambda^{2}\left(I-4\hat{n}+2\hat{n}^{2}\right)\sigma^{z}\nonumber \\
=-\lambda^{2}\left(I-2n\right)\sigma^{z}.
\end{gather}
The gate is therefore directly synthesized as the product of the qubit
rotation gate $\exp\left(-\lambda^{2}\sigma^{z}\right)$ and the $z$-conditional
rotation gate $\exp\left(2\lambda^{2}\hat{n}\sigma^{z}\right)$ for
a lower bound gate depth of two.

\subsection{Effective Hubbard-lattice interaction approach\label{subsec:Effective-Hubbard-Lattice}}

An alternative scheme to encode a qubit in a cavity with the instruction set is to map the qumodes to a qubit by imposing an $\hat{n}\left(\hat{n}-1\right)$
anharmonicity where the Jaynes-Cummings Hamiltonian that describes
the quantum system. This anharmonicity term increases the energy gap between
higher levels of the oscillator to effectively restrict propagation
to the $\text{span}\left\{ \left|0\right\rangle ,\left|1\right\rangle \right\} $
Fock space in which there is universal control.

Consider the standard Jaynes-Cummings Hamiltonian 
\begin{equation}
\hat{H}_{\text{JC}}=\omega_{R}\hat{a}^{\dagger}\hat{a}+\frac{\omega_{Q}}{2}\sigma^{z}+g\left(\hat{a}\sigma^{+}+\hat{a}^{\dagger}\sigma^{-}\right),
\end{equation}
where $\omega_{R}$ is the cavity frequency, $\omega_{Q}$ is the
qubit frequency, and $g$ is the coupling parameter. Inclusion of
the simulated $\hat{n}\left(\hat{n}-1\right)$ anharmonicity of strength
$\Gamma$ yields 
\begin{equation}
\hat{H}_{\text{an}}=\omega_{R}\hat{a}^{\dagger}\hat{a}+\Gamma\hat{n}(\hat{n}-1)+\frac{\omega_{Q}}{2}\sigma^{z}+g(\hat{a}\sigma^{+}+\hat{a}^{\dagger}\sigma^{-}),
\end{equation}
and the system is switched between states $\left|0\right\rangle $ and
$\left|1\right\rangle $ with a weak time-$t$-dependent drive of
strength $\Omega$ at the resonance frequency $\omega_{R}$ as
\begin{equation}
\hat{H}_{\text{drive}}\left(t\right)=\Omega e^{i\omega_{R}t}\hat{a}^{\dagger}+\Omega^{\star}e^{-i\omega_{R}t}\hat{a}.
\end{equation}
Synthesis of a propagator of the form $\exp\left(i\lambda^{2}\hat{n}\left(\hat{n}-1\right)\right)$
is then sufficient to employ the native qumodes as a qubit.
Note the choice of $\lambda$ for practical implementation must take
into account both the time step and the fact BCH
yields a square root in the exponential argument. The required propagator
is given by Trotter for
$\exp\left(\left[\hat{A},\hat{B}\right]\lambda^{2}\right)=\exp(i\lambda^{2}\hat{n}^{2}\sigma^{z})$
and $\exp(-i\lambda^{2}\hat{n}\sigma^{z})$. 

The first term is synthesized according to BCH 
with a commutator determined by the Pauli commutation relation Eq.~(\ref{eq:PauliCommutator})
\begin{align}
\left[\hat{A},\hat{B}\right] & =i\hat{n}^{2}\sigma^{z}\\
 & =i\hat{n}^{2}\left(-\frac{i}{2}\left[\sigma^{x},\sigma^{y}\right]\right)\\
 & =\left[\frac{1}{\sqrt{2}}\hat{n}\sigma^{x},\frac{1}{\sqrt{2}}\hat{n}\sigma^{y}\right],
\end{align}
where $\hat{A}$ and $\hat{B}$ correspond to $x$-conditional and
$y$-conditional rotations, respectively. The second term is a $z$-conditional
rotation gate.

The resulting anharmonicity gate therefore has a gate depth of lower
bound five displacement and rotation gates or 45 displacement gates.

\section{Fermi-Hubbard Lattice Dynamics\label{sec:Fermi-Hubbard-Lattice-Dynamics}}

To further demonstrate the power of our strategy, we employ the
approach to simulate fermionic dynamics on qubit-qumode systems. We detail the required operations, with error estimates an area for future work.
We consider the Fermi-Hubbard lattice Hamiltonian

\begin{align}
\hat{H}_{\text{FH}} & =\hat{T}_{\text{FH}}+\hat{V}_{\text{FH}},\label{eq:FermiHubbardHamiltonian}\\
\hat{T}_{\text{FH}} & =-J\sum_{i,\sigma}\hat{c}_{i,\sigma}^{\dagger}\hat{c}_{i+1,\sigma}+\hat{c}_{i+1,\sigma}^{\dagger}\hat{c}_{i,\sigma},\\
\hat{V}_{\text{FH}} & =U\sum_{i}\hat{n}_{i,\uparrow}\hat{n}_{i,\downarrow}.
\end{align}
The kinetic energy term $\hat{T}_{\text{FH}}$ describes the nearest-neighbor
interaction for hopping of a single spin between two sites with hopping
parameter $J$ and spin $\sigma$ given annihilation operators $\left\{ \hat{c}_{j,\sigma}\right\} $
and creation operators $\left\{ \hat{c}_{j,\sigma}^{\dagger}\right\} $
for sites $\left\{ j\right\} $. The potential energy term $\hat{V}_{\text{FH}}$
describes the same-site interaction, which gives the energetic unfavorability
of a spin up $\uparrow$ and spin down $\downarrow$ coexisting on
the same site $i$, where $\hat{n}_{j,\sigma}$ gives the number of
spin $\sigma$ particles on site $j$. According to fermion statistics,
no more than a single particle of a given spin can exist on a single
site.

Each individual set of qumodes represents either a spin up or spin
down particle on a single lattice site, for direct comparison to the
qubit-based schemes of refs.~\cite{Kivlichan.2018.110501,arute2020observation,Cade.2020.235122}.
Each qmode set is connected to a qumode set that represents the same site
of opposite spin to facilitate computation of the potential energy
$\hat{V}_{\text{FH}}$, as well as to qumode sets of the same spin on neighboring
sites to facilitate computation of the kinetic energy $\hat{T}_{\text{FH}}$.
Qumode sets are also connected along Jordan-Wigner strings to take into
account fermionic statistics. 

The $\text{\ensuremath{\left|0\right\rangle }}$ qumode represents
absence of a spin and the $\text{\ensuremath{\left|1\right\rangle }}$
qumode represents presence of a spin. Within each cavity, only the
qumodes in $\text{span}\left\{ \text{\ensuremath{\left|0\right\rangle }},\text{\ensuremath{\left|1\right\rangle }}\right\} $
are considered, as in Section~(\ref{subsec:Universal-Control}), which
prevents leakage into unphysical high-energy qumodes. At the
end of each operation, the qumode must be in either the $\text{\ensuremath{\left|0\right\rangle }}$
or $\text{\ensuremath{\left|1\right\rangle }}$ mode and the transmon
state must also be in the ground state $\text{\ensuremath{\left|g\right\rangle }}$,
which provides an error syndrome and therefore a degree of error detection 
not employed in qubit-based representations of the Fermi-Hubbard lattice.

Propagation of any combination of up and down spins is simulated
with three gates that operate on two sets of qumodes. The first two gates -- the same-site
and hopping gates -- are defined as the propagator of the same-site
and hopping Hamiltonians, respectively. The same-site term of the
Hamiltonian for site $i$ is 
\begin{equation}
\hat{H}_{\text{same}}=U\hat{n}_{i,\uparrow}\hat{n}_{i,\downarrow}.
\end{equation}
This term is zero if only one spin is on a site and $U$ if both spins
are on the same site, which gives the diagonal Hamiltonian in the
reduced $4\times4$ Hilbert space
\begin{equation}
\hat{H}_{\text{same}}=\left[\begin{array}{cccc}
0 & 0 & 0 & 0\\
0 & 0 & 0 & 0\\
0 & 0 & 0 & 0\\
0 & 0 & 0 & U
\end{array}\right]
\end{equation}
and the diagonal propagator $U_{\text{same}}=\text{e}^{-\text{i}\hat{H}_{\text{same}}\tau}$
\begin{equation}
U_{\text{same}}=\left[\begin{array}{cccc}
1 & 0 & 0 & 0\\
0 & 1 & 0 & 0\\
0 & 0 & 1 & 0\\
0 & 0 & 0 & \text{e}^{-\text{i}U\tau}
\end{array}\right].
\end{equation}
The diagonal propagator is recognized as both the conditional cross-Kerr interaction and the controlled-phase (CPHASE) gate
in the reduced subspace $\text{span}\left\{ \left|0\right\rangle ,\left|1\right\rangle \right\} $.
The hopping term of the Hamiltonian for each $\sigma$ spin in sites
$i,\left(i+1\right)$ is 
\begin{align}
H_{\text{hop}} & =-J\left(\hat{c}_{i,\sigma}^{\dagger}\hat{c}_{i+1,\sigma}+\hat{c}_{i+1,\sigma}^{\dagger}\hat{c}_{i,\sigma}\right)\\
 & =-J\left(\hat{c}_{i,\sigma}^{\dagger}\hat{c}_{i+1,\sigma}-\hat{c}_{i,\sigma}\hat{c}_{i+1,\sigma}^{\dagger}\right),
\end{align}
where the latter expression employs the commutator relationship of
the annihilation and creation operators. The hopping Hamiltonian for
the specified mapping is then the off-diagonal matrix 
\begin{equation}
H_{\text{hop}}=\left[\begin{array}{cccc}
0 & 0 & 0 & 0\\
0 & 0 & -t & 0\\
0 & -t & 0 & 0\\
0 & 0 & 0 & 0
\end{array}\right],
\end{equation}
which gives the hopping propagator $U_{\text{hop}}=\text{e}^{-\text{i}H_{\text{hop}}\tau}$
\begin{align}
U_{\text{hop}} & =\left[\begin{array}{cccc}
1 & 0 & 0 & 0\\
0 & \cos\left(t\tau\right) & i\sin\left(t\tau\right) & 0\\
0 & i\sin\left(t\tau\right) & \cos\left(t\tau\right) & 0\\
0 & 0 & 0 & 1
\end{array}\right],
\end{align}
which is recognized as a conditional controlled-phase beamsplitter
restricted to $\text{span}\left\{ \left|0\right\rangle ,\left|1\right\rangle \right\} $, or equivalently a Givens/iSWAP-like gate in the reduced
$\text{span}\left\{ \left|0\right\rangle ,\left|1\right\rangle \right\} $
subspace \cite{Cade.2020.235122,arute2020observation}. The final
gate of the three-gate set incorporates the fermionic statistics of
the spins via the fermionic SWAP (FSWAP) gate \cite{Kivlichan.2018.110501,Cade.2020.235122}.
The state of each set of qumodes is swapped with one of its neighbors with
inclusion of a phase where both spins are present in neighboring sets
as
\begin{equation}
U_{\text{FSWAP}}=\left[\begin{array}{cccc}
1 & 0 & 0 & 0\\
0 & 0 & 1 & 0\\
0 & 1 & 0 & 0\\
0 & 0 & 0 & -1
\end{array}\right],
\end{equation}
which is recognized as the product of a conditional rotation gate
and a beam-splitter on 3D cQED systems.

Finally, initial states are prepared by the universal set of gates
in $\text{span}\left\{ \left|0\right\rangle ,\left|1\right\rangle \right\} $
detailed in Section~\ref{subsec:Universal-Control}.

\subsection{Conditional cross-Kerr (CPHASE) gate}

We consider the infinitesimal conditional cross-Kerr gate
\begin{equation}
U_{\text{cross-Kerr}}=e^{i\lambda^{2}\hat{n}_{1}\hat{n}_{2}\sigma_{z}},
\end{equation}
which is also employed in Gottesman-Kitaev-Preskill (GKP) codes encoded in qubit-qumode systems \cite{royer2022encoding}. 

The argument is expressed in terms of a commutator according to the
Pauli commutation relation Eq.~\ref{eq:PauliCommutator}
\begin{align}
\left[A,B\right] & \lambda^{2}=i\lambda^{2}\hat{n}_{1}\hat{n}_{2}\sigma_{z}\\
 & =i\lambda^{2}\hat{n}_{1}\hat{n}_{2}\left(-\frac{i}{2}\left[\sigma^{x},\sigma^{y}\right]\right)\\
 & =\left[\frac{1}{\sqrt{2}}\hat{n}_{1}\sigma^{x},\frac{1}{\sqrt{2}}\hat{n}_{2}\sigma^{y}\right]\lambda^{2},
\end{align}
where $\hat{A}$ corresponds to an $x$-conditional rotation gate
and $B$ corresponds to a $y$-conditional rotation gate.

The resulting gate features a lower bound gate depth of four displacement and rotation gates or 16 displacement gates.

\subsection{ \texorpdfstring{$\text{Span}\left\{ \left|0\right\rangle ,\left|1\right\rangle \right\} $}{}
conditional beam splitter gate}

In order to generate a $\text{span}\left\{ \left|0\right\rangle ,\left|1\right\rangle \right\} $
that operates only when $\hat{n}_{1}\hat{n}_{2}\ne1$ (\emph{i.e.},
$1-\hat{n}_{1}\hat{n}_{2}=0$), we formulate the infinitesimal conditional
(controlled-phase) beam-splitter gate
\begin{equation}
U_{\text{cond. beam}}=e^{-i\lambda^{2}\left(\hat{a}_{1}^{\dagger}\hat{a}_{2}+\hat{a}_{1}\hat{a}_{2}^{\dagger}\right)\left(1-\hat{n}_{1}\hat{n}_{2}\right)\sigma^{z}},
\end{equation}
which is decomposed via the Trotter-Suzuki decomposition 
in terms of $\exp\left(-i\lambda^{2}\left(\hat{a}_{1}^{\dagger}\hat{a}_{2}+\hat{a}_{1}\hat{a}_{2}^{\dagger}\right)\sigma^{z}\right)$
and $\exp\left(i\lambda^{2}\left(\hat{a}_{1}^{\dagger}\hat{a}_{2}+\hat{a}_{1}\hat{a}_{2}^{\dagger}\right)\left(\hat{n}_{1}\hat{n}_{2}\right)\sigma^{z}\right)$.
The first term is the conditional beam splitter $U_{\text{beam split.}}$
Eq.~\ref{subsec:Conditional-(Controlled-Phase)-Beam-Splitter} and
the second term is decomposed via BCH 
as described below.

Given the expression of the number operator in terms of the phase-space
operators Eq.~\ref{eq:NumbertoPhaseSpace}, the argument of the second
exponential operator is 
\begin{gather}
i\lambda^{2}\left(\hat{a}_{1}^{\dagger}\hat{a}_{2}+\hat{a}_{1}\hat{a}_{2}^{\dagger}\right)\hat{n}_{1}\hat{n}_{2}\sigma^{z}\nonumber \\
=i\lambda^{2}\left(2\left(\hat{x}_{1}\hat{x}_{2}+\hat{p}_{1}\hat{p}_{2}\right)\right)\hat{n}_{1}\hat{n}_{2}\sigma^{z}.
\end{gather}
The term is then expressed as a Trotter decomposition of $\exp\left(\left[\hat{A}_{1},\hat{B}_{1}\right]\lambda^{2}\right)=\exp\left(2i\lambda^{2}\hat{x}_{1}\hat{x}_{2}\hat{n}_{1}\hat{n}_{2}\sigma^{z}\right)$
and $\exp\left(\left[\hat{A}_{2},\hat{B}_{2}\right]\lambda^{2}\right)=\exp\left(2i\lambda^{2}\hat{p}_{1}\hat{p}_{2}\hat{n}_{1}\hat{n}_{2}\sigma^{z}\right)$. 

The first commutator is given by the Pauli commutation relation Eq.~\ref{eq:PauliCommutator}
\begin{align}
\left[\hat{A}_{1},\hat{B}_{1}\right] & =2i\hat{x}_{1}\hat{x}_{2}\hat{n}_{1}\hat{n}_{2}\sigma^{z}\\
 & =2i\hat{x}_{1}\hat{x}_{2}\hat{n}_{1}\hat{n}_{2}\left(-\frac{i}{2}\left[\sigma^{x},\sigma^{y}\right]\right)\\
 & =\left[\hat{x}_{1}\hat{n}_{1}\sigma^{x},\hat{x}_{2}\hat{n}_{2}\sigma^{y}\right].
\end{align}

The $\hat{A}_{1}$ term is determined by a Trotter decomposition such that
\begin{align}
\hat{A}_{1} & =\frac{1}{2}\left\{ \hat{x}_{1},\hat{n}_{1}\right\} \sigma^{x}+\frac{1}{2}\left[\hat{x}_{1},\hat{n}_{1}\right]\sigma^{x}\\
 & =\hat{A}_{1a}+\hat{A}_{1b},
\end{align}
where, according to the anticommutator-to-commutator relation $\hat{A}_{1a}$,
is given by BCH with 
\begin{align}
\left[\hat{A}_{1a^{\prime}},\hat{B}_{1a^{\prime}}\right] & =\frac{1}{2}\left\{ \hat{x}_{1},\hat{n}_{1}\right\} \sigma^{x}\\
 & =\frac{i}{2}\left[i\hat{x}_{1}\sigma^{y},i\hat{n}_{1}\sigma^{z}\right]\\
 & =\left[-\frac{1}{\sqrt{2}}\hat{x}_{1}\sigma^{y},-\frac{1}{\sqrt{2}}\hat{n}_{1}\sigma^{z}\right],
\end{align}
where $\hat{B}_{1a^{\prime}}$ is a $z$-conditional rotation gate.
Distribution of terms yields $\hat{A}_{1b}$ as 
\begin{equation}
\left[\hat{A}_{1b^{\prime}},\hat{B}_{1b^{\prime}}\right]=\left[\frac{1}{\sqrt{2}}\hat{x}_{1},\frac{1}{\sqrt{2}}\hat{n}_{1}\sigma^{x}\right],
\end{equation}
where $B_{1b^{\prime}}$ is an $x$-conditional rotation gate. 

According to the same procedure,
\begin{align}
\hat{B}_{1} & =\frac{1}{2}\left\{ \hat{x}_{2},\hat{n}_{2}\right\} \sigma^{y}+\frac{1}{2}\left[\hat{x}_{2},\hat{n}_{2}\right]\sigma^{y}\\
 & =\hat{B}_{1a}+\hat{B}_{1b},
\end{align}
where $\hat{B}_{1a}$ is given by
\begin{align}
\left[\hat{A}_{1a^{\prime\prime}},\hat{B}_{1a^{\prime\prime}}\right] & =\frac{1}{2}\left\{ \hat{x}_{2},\hat{n}_{2}\right\} \sigma^{y}\\
 & =\frac{i}{2}\left[i\hat{x}_{2}\sigma^{z},i\hat{n}_{1}\sigma^{x}\right]\\
 & =\left[-\frac{1}{\sqrt{2}}\hat{x}_{2}\sigma^{z},-\frac{1}{\sqrt{2}}\hat{n}_{2}\sigma^{x}\right],
\end{align}
with $B_{1a^{\prime\prime}}$ an $x$-conditional rotation, and $B_{1b}$
is given by 
\begin{equation}
\left[\hat{A}_{1b^{\prime\prime}},\hat{B}_{1b^{\prime\prime}}\right]=\left[\frac{1}{\sqrt{2}}\hat{x}_{2},\frac{1}{\sqrt{2}}\hat{n}_{2}\sigma^{y}\right],
\end{equation}
where $B_{1b^{\prime\prime}}$ is a $y$-conditional rotation gate.
The second term follows analogously with the position $x$ replaced
by the momentum $p$.

\subsection{Conditional FSWAP gate}

The FSWAP gate follows immediately from the conditional
cross-Kerr gate detailed above and a complete beam-splitter gate (or
conditional beam-splitter gate detailed above) as 
\begin{equation}
U_{\text{FSWAP}}=U_{\text{cond. Kerr}}U_{\text{cond. beam}}.
\end{equation}

\newpage
\printbibliography
\end{document}